\documentclass[reprint,amsfonts, amssymb, amsmath,  showkeys,pra,nofootinbib, twocolum,,longbibliography,aps, superscriptaddress]{revtex4-2}

\usepackage{float}
\makeatletter
\let\newfloat\newfloat@ltx
\makeatother
\usepackage[english]{babel}
\usepackage[utf8]{inputenc}
\usepackage{graphics}
\usepackage{selinput}
\usepackage[normalem]{ulem}
\usepackage[shortlabels]{enumitem}
\usepackage{braket}
\usepackage{amsthm}
\usepackage{mathtools}
\usepackage{physics}
\usepackage[dvipsnames]{xcolor}
\usepackage{graphicx}
\usepackage[left=16mm,right=16mm,top=35mm,columnsep=15pt]{geometry} 
\usepackage{adjustbox}
\usepackage{placeins}
\usepackage[T1]{fontenc}
\usepackage{lipsum}
\usepackage{csquotes}
\usepackage{bm}
\usepackage[linesnumbered,ruled,vlined]{algorithm2e}
\SetKwInput{kwInit}{Init}
\usepackage{changes}
\definechangesauthor[name=Ilya Safro, color=red]{is}
\usepackage[makeroom]{cancel}
\usepackage[toc,page]{appendix}
\usepackage[colorlinks=true,citecolor=blue,linkcolor=magenta]{hyperref}
\usepackage{tikz}
\tikzset{every picture/.style=remember picture}
\usepackage{graphicx}

\usepackage{enumitem}

\usepackage{dcolumn}
\usepackage{bm,bbm}
\usepackage{physics}
\usepackage{amsmath}
\usepackage{amsthm}
\usepackage{amsfonts}
\usepackage{hyperref}
\usepackage{xcolor}
\usepackage{float}
\usepackage{soul}
\sethlcolor{yellow}

\usepackage{lipsum}

\usepackage{cancel}

\newcommand{\1}{\mathbbm{1}}

\newcommand{\sx}{\sigma_x}
\newcommand{\sy}{\sigma_y}

\newcommand{\fsnull}[1]{}
\newcommand{\old}[1]{}

\def\LC{\mathcal{L}}

\usepackage{mathrsfs}

\DeclareMathOperator{\sinc}{sinc}
\DeclareMathOperator{\Ebb}{\mathbb{E}}

\newcommand{\DC}{\mathcal{D}}

\newcommand{\VC}{\mathcal{V}}

\renewcommand{\geq}{\geqslant}
\renewcommand{\leq}{\leqslant}

\DeclareMathOperator*{\argmin}{arg\,min}

\DeclareMathOperator{\Var}{{\rm Var}}
\DeclareMathOperator{\Covar}{{\rm Cov}}

\renewcommand{\vec}[1]{\boldsymbol{#1}}

\newcommand*{\id}{\openone}

\newcommand{\bs}{\textsf{BS}}

\newcommand{\al}{\alpha }

\renewcommand{\th}{\theta }

\newcommand{\thv}{\vec{\theta}}
 
\newcommand{\alv}{\vec{\alpha} }

\def\be{\begin{equation}}
\def\ee{\end{equation}}
\def\bs{\begin{split}}
\def\e{\end{split}}
\def\ba{\begin{eqnarray}}
\def\bea{\begin{eqnarray}}

\def\tea{\end{eqnarray}}
\def\ea{\end{eqnarray}}
\def\eea{\end{eqnarray}}

\newtheorem{theorem}{Theorem}

\newtheorem{proposition}{Proposition}

\newtheorem{definition}{Definition}

\usepackage{dsfont}

\def\be{\begin{equation}}
\def\te{\end{equation}}
\def\ee{\end{equation}}
\def\ba{\begin{eqnarray}}
\def\bea{\begin{eqnarray}}

\def\tea{\end{eqnarray}}
\def\ea{\end{eqnarray}}
\def\eea{\end{eqnarray}}

\newcommand{\beq}{\begin{equation}}
\newcommand{\eeq}{\end{equation}}

\newcommand{\paramh}{x}

\makeatletter

\newif\ifeqcontrib@this
\newif\ifeqcontrib@any
\eqcontrib@thisfalse
\eqcontrib@anyfalse

\newcommand{\eqcontrib}{\global\eqcontrib@thistrue\global\eqcontrib@anytrue}

\newcommand{\eqcontribmark}{\textsuperscript{\ensuremath{,\ddagger}}}

\newcommand{\eqcontrib@maybe}{%
  \ifeqcontrib@this
    \global\eqcontrib@thisfalse
    \eqcontribmark
  \fi
}

\newcommand{\printEqContrib}{%
  \ifeqcontrib@any
    \begingroup
      \renewcommand{\thefootnote}{\ensuremath{\ddagger}}%
      \footnotetext{These authors contributed equally to this work.}%
    \endgroup
  \fi
}

\def\doauthor#1#2#3{%
  \ignorespaces#1\unskip\@listcomma
  \begingroup
    #3%
  \@if@empty{#2}{\endgroup{}{}}{\endgroup{\comma@space}{}\frontmatter@footnote{#2}}%
  \eqcontrib@maybe
  \space \@listand
}%
\makeatother

\begin{document}

\preprint{APS/123-QED}

\title{Warm Starts, Cold States: Exploiting Adiabaticity for Variational Ground-States}

\author{Ricard Puig\eqcontrib}
\email{ricard.puigivalls@epfl.ch}

\affiliation{Institute of Physics, Ecole Polytechnique F\'{e}d\'{e}rale de Lausanne (EPFL), CH-1015 Lausanne, Switzerland}
\affiliation{Centre for Quantum Science and Engineering, Ecole Polytechnique F\'{e}d\'{e}rale de Lausanne (EPFL), CH-1015 Lausanne, Switzerland}

\author{Berta Casas\eqcontrib}
\email{berta.casas@bsc.es}

\affiliation{Barcelona Supercomputing Center, Plaça Eusebi G\"uell, 1-3, 08034 Barcelona, Spain}
\affiliation{Universitat de Barcelona, 08007 Barcelona, Spain}

\author{Alba Cervera-Lierta}
\affiliation{Barcelona Supercomputing Center, Plaça Eusebi G\"uell, 1-3, 08034 Barcelona, Spain}

\author{Zo\"e Holmes}
\affiliation{Institute of Physics, Ecole Polytechnique F\'{e}d\'{e}rale de Lausanne (EPFL), CH-1015 Lausanne, Switzerland}
\affiliation{Centre for Quantum Science and Engineering, Ecole Polytechnique F\'{e}d\'{e}rale de Lausanne (EPFL), CH-1015 Lausanne, Switzerland}

\author{Adri\'an P\'erez-Salinas}
\affiliation{Institute for Theoretical Physics, ETH Zürich, 8093, Zürich, Switzerland}

\date{\today}

\begin{abstract}
Reliable preparation of many-body ground states is an essential task in quantum computing, with applications spanning areas from chemistry and materials modeling to quantum optimization and benchmarking. A variety of approaches have been proposed to tackle this problem, including variational methods. However, variational training often struggle to navigate complex energy landscapes, frequently encountering suboptimal local minima or suffering from barren plateaus. In this work, we introduce an iterative strategy for ground-state preparation based on a stepwise (discretized) Hamiltonian deformation. By complementing the Variational Quantum Eigensolver (VQE) with adiabatic principles, we demonstrate that solving a sequence of intermediate problems facilitates tracking the ground-state manifold toward the target system, even as we scale the system size. We provide a rigorous theoretical foundation for this approach, proving a lower bound on the loss variance that suggests trainability throughout the deformation, provided the system remains away from gap closings. Numerical simulations, including the effects of shot noise, confirm that this path-dependent tracking consistently converges to the target ground state.
\end{abstract}

\maketitle
\printEqContrib

\section{Introduction}

Learning ground states and their energies remain among the most compelling uses of quantum computing, with applications in electronic-structure calculations~\cite{cao2019quantum}, large-scale optimization~\cite{farhi2014quantum}, high energy physics~\cite{di2024quantum}, and quantum sensing~\cite{hassani2025many, puig2024dynamical}. Yet a persistent bottleneck underlies essentially all known approaches to ground-state preparation and ground-state energy estimation~\cite{nielsen2000quantum,temme2011quantum,kitaev1995quantum,brassard2000quantum,tan2020quantum,low2019hamiltonian,gilyen2019quantum,ge2019faster,motta2020determining,motlagh2024ground,gluza2024double, suzuki2025double}:
they typically require initialization with a state that already has nontrivial overlap with the true ground state. How to prepare such “good” initial states is currently the subject of active debate.

Classical experience provides insights into how one might approach approximate state preparation. One natural strategy is to optimize within a structured variational family; for example, density matrix renormalization group (DMRG) methods form a bedrock of electronic-structure calculations. This perspective has, in turn, motivated variational quantum methods such as the Variational Quantum Eigensolver (VQE), a hybrid quantum–classical algorithm that minimizes an energy objective of a parameterized quantum state preparation circuit. A second common strategy is to exploit adiabaticity—slowly deforming the ground state of an easy-to-solve Hamiltonian into that of a harder one~\cite{garcia2018addressing,harwood2022improving,albash2018adiabatic, hibat2021variational, vzunkovivc2025variational}. This principle underlies quantum adiabatic algorithms for ground-state preparation (in both analogue~\cite{farhi2000quantum} and digital settings~\cite{kolotouros2024simulating}) where one evolves under a slowly varying Hamiltonian and leverages a spectral gap to suppress unwanted excitations.

\begin{figure*}[t!]
    \centering
\includegraphics[width=1\linewidth]{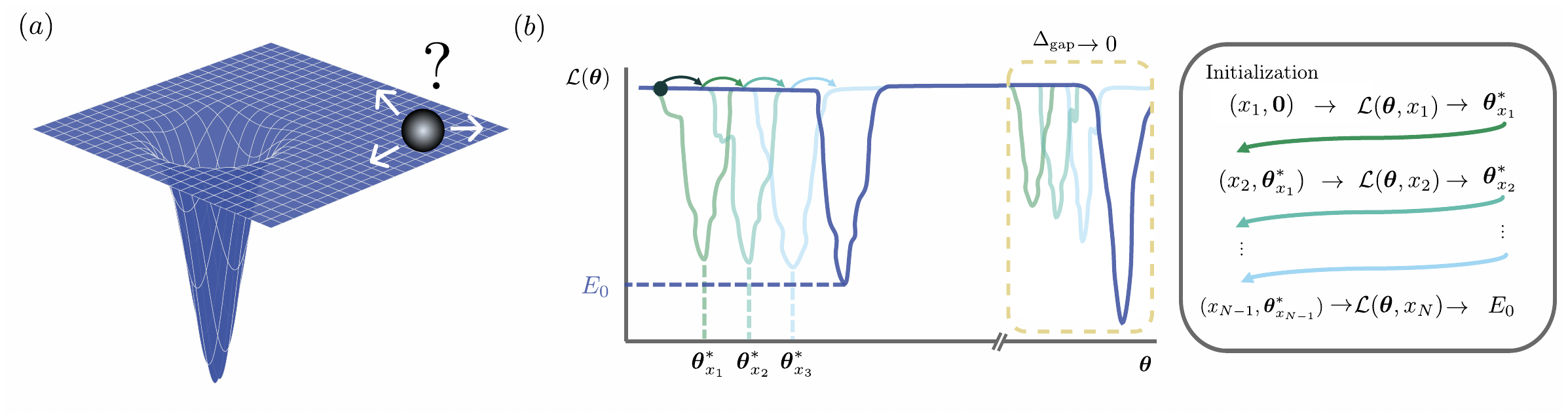}
    \caption{
    \textbf{Schematic overview of the iterative training mechanism and its failure mode at level crossings.} \textbf{(a)} Illustrative loss landscape at random initialization: the minimum is narrow and the landscape is otherwise nearly flat, so a random start typically lies in a low-gradient region (barren plateau). \textbf{(b)} \textbf{Left:} The loss landscape $\mathcal{L}(\thv)$ as a function of the variational parameters $\thv$ at successive iterations (colors). When the ground state varies smoothly and the spectral gap remains open, the location of the minimum drifts continuously, so initializing each step with the previous optimum keeps the optimizer in a high-gradient region that reliably tracks the moving solution. As the gap closes (yellow dashed box), a competing global minimum can emerge far from the followed basin, causing the iterative procedure to remain trapped in the “wrong” minimum past the crossing. \textbf{Right:} visual representation of the iterative warm-start strategy, where the optimal parameters at iteration $k$ are used to initialize the step $k+1$. }
    \label{fig:mainfig}
\end{figure*}

In this work, we develop an iterative training strategy that integrates adiabatic state preparation~\cite{farhi2000quantum, albash2018adiabatic} into VQE~\cite{peruzzo2014variational}. Starting from a trivial Hamiltonian with known parameters that prepare the ground state, we gradually deform the Hamiltonian along an adiabatic interpolation. At each step, the variational parameters are initialized in the optimal solution from the preceding Hamiltonian in the path, and then optimized. We further generalize this approach from vanilla VQE, which targets the ground state of a single Hamiltonian, to meta-VQE~\cite{cervera2020meta}, in which a single parameterized model is trained to output near-optimal parameters across a family of Hamiltonians. We stress that in both cases our iterative strategy is not limited to quantum hardware in that it could also be applied to ``quantum-inspired'' classical methods that rely on simulating quantum circuits or variational states to find ground states and energies~\cite{lerch2024efficient,rudolph2025pauli}.

The iterative training strategy proposed here mitigates key limitations of VQE and adiabatic state preparation when each is used in isolation. Compared to adiabatic methods, an iterative variational scheme can potentially prepare similar quality states with substantially shallower circuits, and it provides a natural setting in which to learn and implement shortcuts to adiabaticity~\cite{guery2019shortcuts}. Relative to standard VQE, the method is designed to improve trainability in the presence of barren-plateau phenomena: by exploiting the smooth deformation of the Hamiltonian and its objective landscape, each step can be warm-started from the preceding solution, keeping the optimization localized and the variational search in a favorable regime.

To formalize the benefits of warm starts, we derive a lower bound on the gradient variance, showing that at each step of the iterative procedure both VQE and Meta-VQE maintain trainable gradients, provided the preceding step converges sufficiently. Notably, we find that these guarantees are sensitive to the spectral gap and fail to hold when it vanishes.
Such breakdowns are unsurprising:
finding ground states of general many-body Hamiltonians is QMA-complete~\cite{kempe2004complexity}.

To investigate the practicality of our method, we supplement our theoretical results with numerical simulations of the full training process, accounting for realistic shot noise. Our results confirm that the method consistently converges 
when the gap in the ground state energy does not vanish, providing a promising path for ground-state preparation on near-term devices in such cases. However, we see that when the gap closes the method can break down in line with our theoretical predictions and the QMA-completeness of generic Hamiltonians.

\section{Preliminaries}\label{sec:prelimianries}

\subsection{Variational Quantum Eigensolver and Meta-VQE}

The Variational Quantum Eigensolver (VQE) is a hybrid quantum–classical variational algorithm that approximates the ground-state energy of a Hamiltonian $H_P$ by preparing a trial state on a quantum processor and minimizing its measured energy over a parametrized circuit family. Precisely, the algorithm optimizes the following loss function
\begin{align}\label{eq:lc_vqe}
    \LC_{\rm VQE}(\thv) =\expval{U^\dagger(\thv) H_P U(\thv)}{\psi},
\end{align}
where $U(\thv)$ is a parametrized quantum circuit, $\ket{\psi}$ is an initial state, and $\thv$ is the vector of variational parameters. These parameters are optimized via a hybrid quantum-classical loop: the quantum computer estimates the objective and (when required) its gradients, which are later fed to a classical optimizer suggesting a new set of parameters. The loop is repeated until convergence to a candidate solution $\thv^*$.
In this case, we focus on a very general circuit of the form
\begin{align}\label{eq:circuit_vqe}
    U_{\rm VQE}(\thv) = \prod_{j=1}^M V_j e^{-i\th_j P_j}\,,
\end{align}
where $\{P_j\}_{j=1}^{M}$ is a set of Hermitian generators such that $P_j^2 = \1$, $M$ is the total number of variational parameters, and  $\{V_j\}_{j=1}^M$ are fixed (non-parametrized) gates. Throughout this work, we assume the parameters $\{\th_j\}$ to be uncorrelated.

A very natural extension of VQE for parametrized Hamiltonians is the Meta-VQE~\cite{cervera2020meta}. The Meta-VQE is an algorithm that contains an additional input (or encoding parameter) $x$ in the circuit, and is trained on a family of parametrized Hamiltonians $\{H_P(x_i)\}_{i=1}^N$. The loss function is the average energy over a training set $\{x_i\}_{i = 1}^N$,
\begin{align}\label{eq:lc_metavqe}
    \LC_{\rm MVQE}(\thv) = \frac{1}{N}\sum_{i=1}^N \expval{U^\dagger(\thv,x_i)H_P(x_i)U(\thv, x_i)}{\psi}\,.
\end{align}
Minimizing $ \LC_{\rm MVQE}(\thv)$ results in a single parameter vector $\thv$ that can be used with the same circuit family $U(\thv,x)$ to prepare approximate ground states for all training points $x_i$.  
A central motivation is that the learned parameters may also transfer to previously unseen values of $x$, for example between the trained points.

Because the circuit explicitly depends on $x$, we write it as
\begin{align}\label{eq:circuit_metavqe}
    U_{\rm MVQE}(\thv,x) = \prod_{j}^MV_j e^{-if_j(\th_j,x) P_j}\,,
\end{align}
where the notation carries over from Eq.~\eqref{eq:circuit_vqe}. The encoding function $f_j(\th_j,x)$ specifies how the input $x$ modulates the $j$-th rotation. In this sense, the choice of encoding determines how the data (here, the Hamiltonian parameter $x$) is injected into the circuit, and can strongly influence the induced distribution of states and the resulting optimization  of the model~\cite{casas2024role}.  For the purpose of this work we will consider $f_j(\th_j,x)$ to be a linear function of the training parameters, i.e., $f_j(\th_j,x)  = g_j(x) \th_j$, where each $g_j(x)$ has bounded derivatives in the domain of $[{x}_{\min},{x}_{\max}]$. 

\subsection{Gradient magnitudes and barren plateaus}
Significant effort has been devoted to understanding the trainability of VQE and other Variational Quantum Algorithms (VQAs). In practice, training can be hindered by limited expressivity, unfavorable optimization landscapes, and local minima. One of the most studied phenomena is the vanishing gradient problem, also known as the barren plateau problem~\cite{mcclean2018barren, larocca2024review}, which is typically characterized by an exponentially vanishing variance of the cost landscape (over an ensemble of parameters or circuits). This implies that most of the parameter settings in the landscape yield cost values that are very close to the ensemble average, by Chebyshev's inequality
\begin{align}
    {\rm Pr}(|\LC(\thv)- \Ebb_{\thv\sim\DC}[\LC(\thv)]|\geq\epsilon)\leq \frac{{\rm Var}_{\thv\sim\DC}[\LC(\thv)]}{\epsilon^2}.
\end{align}
If the fluctuations are exponentially small, exponentially many physical measurements are required, making optimization inefficient. 
Relatedly, gradients can also exhibit exponential concentration \cite{perez-salinas2024analyzing, arrasmith2021equivalence}.
More broadly, these effects can be viewed as manifestations of a curse of dimensionality~\cite{cerezo2023does}.
This has prompted the search for specific ansätze in which the variance vanishes at most polynomially with the number of qubits.

\subsection{Iterative methods}
Iterative methods combat trainability barriers by breaking down a single hard optimization task into a sequence of smaller, more manageable subproblems.  
The key idea is that the optimal solutions vary smoothly across the sequence, so that the parameters optimized at one step provide a good warm-start initialization for the next.

One way to generate such a sequence is to introduce an interpolation between an ``easy'' initial Hamiltonian $H_{\mathrm{ini}}$ with known ground state and the target problem Hamiltonian $H_P$, similar to an adiabatic path~\cite{farhi2000quantum}
\begin{align}
    H(s) = (1-s)H_{\mathrm{ini}} + s H_P,
\end{align}
As $s$ increases from $0$ to $1$, this defines a continuous path between an easy instance to the target instance, which mirrors the spirit of adiabatic state preparation. Unlike true adiabatic evolution, however, the state is produced by a parametrized circuit and updated through variational optimization at discrete values of ${s_k}$. This has two main consequences: on the one hand, success depends on the expressivity of the ansatz and the ability to remain on the relevant low-energy branch. On the other hand, the depth required to learn and prepare the ground states can potentially be significantly shallower in a digital setting than simulating the adiabatic time evolution.

In the context of many-body Hamiltonians, it is often more natural to consider families parameterized by a physical coupling $x$, for example
\begin{equation}\label{eq:param_H_iterative_general}
    H(x) = H_0 + x H_1,
\end{equation}
For any fixed $x$, the goal is to approximate the ground state of the corresponding instance $H(x)$ (which we denote by $H$ when $x$ is clear from context).
By incrementally changing the value of $x$, we can create a set of optimization problems that are similar enough. Concretely, we choose a set $x_1<x_2<\cdots<x_K$ and define a family of loss functions $\{\LC(\thv;x_k)\}_{k = 1}^K$.
The optimal parameters $\thv^*_{x_{k-1}}$ found at step $k-1$ are expected to provide a useful inductive bias for step $k$. As we will explore, this depends on several factors, such as the gap of the Hamiltonian, the smoothness of the encoding function, and the circuit expressivity. In practice, we initialize the parameters for the $k$-th problem by sampling them randomly from a small region (for instance, a hypercube) centered at $\thv^*_{x_{k-1}}$.

We formalize this idea in the following definition.
\begin{definition}[Iterative training strategy]\label{def:iterative_training_strat}
Consider a general loss function $\LC(\thv)$. An iterative training strategy defines a family of loss functions $\{\LC(\thv, x_k)\}_{k=1}^K$, and uses the solution at step $k-1$
\begin{equation}
    \thv^*_{x_{k-1}}= \argmin_{\thv}\LC(\thv, x_{k-1})\,,
\end{equation}
as an initialization bias for optimizing $\LC(\thv, x_{k})$. Concretely, the initial parameters for $\LC(\thv,x_{k})$ are chosen according to a distribution  $\DC$ peaked at $\thv^*_{x_{k-1}}$.  
\end{definition}

For a VQE loss function of the form of Eq.~\eqref{eq:lc_vqe}, the family of loss functions is defined as 
\begin{align}
     \LC_{\rm VQE}(\thv, x_k) =\expval{U^\dagger(\thv) H(x_k) U(\thv)}{\psi},
\end{align}
with $H(x_k) = H_0 + x_k H_1$ as in Eq.~\eqref{eq:param_H_iterative_general}. Importantly, $\LC(\thv,x_K)$ corresponds to the original problem we wanted to solve $\LC(\thv)$. Similarly, for a Meta-VQE loss function of the form of Eq.~\eqref{eq:lc_metavqe}, the family of loss functions that we define are $\{\LC_{\rm MVQE}(\thv, k)\}_{k = 1}^{K}$, where
\begin{multline}
   \LC_{\rm MVQE}(\thv,k) = \\ \frac{1}{k}\sum_{j=1}^k\expval{U^\dagger(\thv,x_j)H(x_j)U(\thv,x_j)}{\psi}\,.
\end{multline}
Notice the following key difference between the Meta-VQE and the simple VQE iterative strategy: in the Meta-VQE, instead of changing each summand in the loss function, an extra term is added to it, leaving us with sequence of objective functions indexed by the number of training points.

\section{Main results}\label{sec:results}

\subsection{Overview}
Barren plateaus are most commonly established by showing that the cost function concentrates under a chosen parameter ensemble, often corresponding to random initialization. Importantly, this is an average-case statement: it does not exclude the existence of trainable regions of the landscape.
In particular, large enough gradients typically exist if the optimization is initialized closely enough to a minimum~\cite{mhiri2025unifying}. While we generally cannot guarantee initialization in such a neighborhood a priori, we can devise algorithms and techniques to find good starting points with provable guarantees. In particular, it has been shown that small angle initializations can \textit{sometimes} avoid small gradients.

To formalize these results, let us define a hypercube in the parameter space centered at a reference point $\thv^* \in \mathbb{R}^M$ and with side $2r$ as
\begin{align}
    \VC(\thv^*,r) = \{\thv = \thv^* + \alv\,|\,\al_i\in[-r,r]\}\,.
\end{align}
Then, we we can define a localized distribution. Concretely, we denote $\DC(\thv^*,r)$ the uniform distribution over the hypercube $\VC(\thv^*,r)$, i.e., the distribution obtained by drawing each $\al_i$ independently and uniformly from $[-r,r]$.

Prior work showed that initializing a circuit uniformly around $\thv^*= \vec{0}$ with $r\in\order{1/\sqrt{L}}$, where $L$ is the number of layers, can lead to non-exponentially vanishing gradients for a specific hardware and initial state~\cite{wang2023trainability}. Similar conclusions were reached in~\cite{dborin2022matrix, truger2024warm,rudolph2022synergy, goh2023lie,sauvage2021flip, verdon2019learning, okada2023classically, ravi2022cafqa,gibbs2024exploiting, mitarai2022quadratic, tate2021classically, niu2023warm, egger2021warm, wurtz2021fixed, mari2020transfer, wilson2019optimizing, liu2023mitigating,zhou2020quantum, akshay2021parameter, grimsley2022adapt, mele2022avoiding}, and most of these results have been unified and generalized in~\cite{mhiri2025unifying}, where it is shown that gradients vanish at worst polynomially with the number of qubits around any reference point with non-exponentially vanishing second derivative.  However, since identity initialization can fail to have gradients in full generality, the work conjectures that warm starts need to get increasingly \textit{smarter} with the number of qubits.

Variational methods have also been used to compress Trotter circuits via warm starts~\cite{puig2024variational} by learning one Trotter step at a time, using the previous step as an inductive bias. More concretely, the authors use established methods to iteratively compress a Trotter evolution, and rigorously prove that at each new initialization gradients can be found. This concept is generalized to any overlap type loss, with a rank-1 target state. They show that for any target state that changes sufficiently slow, we can use iterative methods to enhance the training of a VQA. Crucially, this work does not apply to loss functions with observables with rank larger than one. Or in other words, the method fails for ground state preparation. 

This critical limitation motivate us to explore \textit{iterative training} beyond loss functions that are just the overlap between two states, and to study whether iterative approaches can be used to find ground states of arbitrary observables. In particular, to approximate the ground state of $H(x)$, we propose to split the problem into an iterative sequence of easier optimization tasks. Indeed, by starting our training in a reference Hamiltonian $H_0$, whose ground state is known or easy to find, and incrementally modifying it until we reach the final $H(x)$, we can potentially remain in regions with non-exponentially vanishing gradients. 
 
Concretely, in this work we provide analytical guarantees for iterative training methods covering a general family of quantum circuits and variational algorithms. The guarantees are based on the fact that, by finding a good approximation to the ground state of an intermediate step, we can use it as an effective warm-start initialization for the next iteration, as long as the change in $H(x)$ is \textit{small enough}. This procedure is sketched in Fig.~\ref{fig:mainfig}. 

Our results can be viewed as a discretized analogue of adiabatic quantum computation in the following sense. In adiabatic state preparation, one follows a continuous path of Hamiltonians $H(s)$ and, provided the evolution is slow compared to the inverse spectral gap, the system remains close to the instantaneous ground state. Here, instead of implementing a continuous-time evolution, we consider a \emph{sequence} of Hamiltonians along an interpolation path and solve a variational optimization problem at each point. The underlying mechanism is analogous to adiabatic reasoning: if successive Hamiltonians are sufficiently close and the gap does not become too small, the ground state varies smoothly along the path, so the parameters learned at step $k-1$ provide an effective warm start for step $k$. Our analysis makes this intuition quantitative by translating the notion of ``slow enough'' into explicit, step-wise conditions expressed in terms of how much the Hamiltonian changes between consecutive points on the path. In this way, we provide analytical guarantees for when iterative training remains in a trainable regime and when it can fail, for example when gap closures or level crossings cause the optimum to move abruptly from one step to the next.

\subsection{Lower-bound on the variance}\label{sec:lower_bound_main}
We now proceed to analyzing iterative strategies through the local geometry of the loss landscape around the warm-start initialization. At iteration $x_k$, we do not optimize from scratch; instead, we initialize parameters in a small neighborhood of the previously optimized solution $\thv^*_{x_{k-1}}$. The key question is then whether this \emph{local} region remains sufficiently informative for optimization as we move along the path.

To move away from the intuition and turn this into a formal result, in this section we provide a lower bound on the variance of the loss function at a new iteration $x_k$ when parameters are sampled from a hypercube around $\thv^*_{x_{k-1}}$. This bound serves as a quantitative certificate that the loss does not become overly flat in the warm-start region. Equivalently, it guarantees that typical parameter perturbations around $\thv^*_{x_{k-1}}$ induce non-negligible changes in the objective, which is a necessary ingredient for reliable training and for assessing progress with a feasible number of measurements.

\begin{theorem}[Lower-bound on the loss variance for iterative methods, Informal]\label{thm:main}
    Consider a loss function of the form in either Eq.~\eqref{eq:lc_vqe} or Eq.~\eqref{eq:lc_metavqe}, with \textit{ans\"atze} as defined in Eq.~\eqref{eq:circuit_vqe}, and Eq.~\eqref{eq:circuit_metavqe} with $M$ variational parameters. Suppose we use an iterative training strategy according to Definition~\ref{def:iterative_training_strat}. Assume that the following conditions are fulfilled:
    \begin{enumerate}
        \item The Hamiltonian change rate, $x_k-x_{k-1}$ is bounded by
         \begin{align}
            |x_k-x_{k-1}|\in \order{\frac{|\Delta_{\rm gap}|}{\|H_1\|_s + M\|H(x_k)\|_s}}\,.
        \end{align}
         \item  The infidelity between the true ground state and the one learned variationally in the previous iteration is bounded by a constant.
         \item The generator of the first gate acts non-trivially on the initial state, i.e., $\expval{P_1}{\psi} = 0$.
    \end{enumerate}
   
   Consider sampling the parameters $\thv_{k}$ from a hypercube of width $2r$ around the previous solution $\thv_{k-1}^*$ such that 
    \begin{align}
        r^2\in\order{\frac{|\Delta_{\rm gap}|}{M (\|H(x_k)\|_s + |\Delta_{\rm gap}|)}}\,.
    \end{align}
    Then we can guarantee the variance lower bound
    \begin{align}
        {\rm Var}_{\thv\sim\DC(\thv^*,r)}[\LC(\thv, x_k)]\in\Omega\left(r^4\right)\,.
    \end{align}
    \end{theorem}
Here we used $\|\cdot\|_s = \lambda_{\max}^{(\cdot)} -  \lambda_{\min}^{(\cdot)}$, a semi-norm~\cite{boixo2007generalized} that captures the scale of a Hamiltonian, where $\lambda_{\max}^{(\cdot)}$ and $\lambda_{\min}^{(\cdot)}$ denote, respectively, the largest and smallest eigenvalues of the corresponding Hamiltonian. Furthermore, we defined $|\Delta_{\rm gap}| = |\lambda_{\min}^{(\cdot)} - \lambda_{1}^{(\cdot)}|$, the absolute difference between the ground state and the first eigenvalue. 
The formal version of Theorem~\ref{thm:main} is presented and proven in Appendix~\ref{app:noencoding} for the VQE and Appendix~\ref{app:lowerbound_metavqe} for the Meta-VQE algorithm. For clarity, we present in the main text a single unified theorem that covers both cases by adopting the more restrictive set of assumptions needed across the two analyses.

Theorem~\ref{thm:main} provides a step-wise guarantee that, around each warm-start point, the loss landscape does not become exponentially flat. Concretely, it establishes that inside a hypercube of side length $2r$, with $r^2\propto |\Delta_{\rm gap}|/(M\|H\|_s)$, the variance decays at worst polynomially in $r$. Consequently, as long as the gap $|\Delta_{\rm gap}|$ does not close along the path, and as long as $\|H\|_s$ and $M$ do not scale exponentially with the number of qubits, the variance decays polynomially, ensuring a local region with provable gradients along the iterative path.

The rate at which the Hamiltonian is iteratively updated, i.e., $|x_k-x_{k-1}|$,  also needs to be sufficiently gradual to guarantee a non-exponentially vanishing variance (and hence a non-flat loss landscape). That said, the first condition in Theorem~\ref{thm:main}, $|x_k-x_{k-1}|\in |\Delta_{\rm gap}|/(\|H_{1}\|_s + M \|H(x_{k-1})\|_s)$, is obtained from a loose, worst-case bound and may be improvable in practice. In contrast, the dependence on the gap is structural: as we exemplify in Section~\ref{sec:training}, the gap closing can render the training infeasible. Indeed, the closure of the gap in the ground states acts as a roadblock, similar to what has been found in other contexts~\cite{farhi2000quantum, albash2018adiabatic}. Thus, the scaling with the gap is unlikely to be improved.

The second condition in Theorem~\ref{thm:main} captures the intuitive fact that the training at each iteration needs to be sufficiently successful to ensure that, at the next iteration, we are close enough to a minimum to have gradients. If the training of the loss function $k-1$ fails (or does not converge well enough), the subsequent warm start may fall outside the guaranteed region, which can cause the iterative training scheme to be ineffective.

The final condition to guarantee loss variances in Theorem~\ref{thm:main} is that the first gate in our circuit acts non-trivially on the initial state (e.g. does not commute entirely with the target Hamiltonian). We impose $\expval{P_1}{\psi} = 0$ as a convenient sufficient condition to simplify the proof, but this requirement might not be strictly necessary. Its role is simply to ensure that the ansatz produces a nontrivial local variation of the prepared state (and hence of the loss) around the initialization.

We devote the rest of this section to studying the tightness of the bound with respect to the number of parameters $M$.
Fig.~\ref{fig:variance_scaling} analyzes the local geometry of the loss landscape $\LC_{\rm VQE}(\thv)$ for a single iterative step $k$ of the training procedure. Concretely, we estimate the variance of the loss under random perturbations of the optimized parameters $\thv^*(x_k)$, which serve as a warm start for the next iteration of the algorithm. This characterizes how loss fluctuations concentrate as a function of the neighborhood size $r$ and the parameter count $M$ (linear in the system size $n$), providing guidance on how far one can move in $x$ while maintaining usable gradient information when $n$ grows.

In our simulations, we consider a one-dimensional Heisenberg Hamiltonian with a transversal magnetic field of the form:
\begin{multline}
    H({x})
    = - \sum_{i=1}^{n} Z_i + x \sum_{\langle i,j\rangle}
\bigl(
X_i X_j + Y_i Y_j + Z_i Z_j
\bigr),
    \label{eq:Ising_Ham}
\end{multline}
where the first term corresponds to a gapped Hamiltonian with a product-state ground state, and the second term defines a Heisenberg interaction Hamiltonian with $\langle i,j\rangle$ denoting nearest neighbor interactions with periodic boundary conditions. The interpolation parameter $x$ therefore connects an easy-to-prepare initial regime to a strongly interacting regime as $x$ is increased. 

Following the ansatz definition in Eq.~\eqref{eq:circuit_vqe}, we use a circuit of the form $U(\boldsymbol{\theta}) = \prod_{l=1}^{L}U_l(\boldsymbol{\theta}_l)$, where each layer $U_l(\boldsymbol{\theta}_l)$ consists of single-qubit $R_z$ rotations on all qubits, followed by two-qubit rotations $R_{xx}$, $R_{yy}$, and $R_{zz}$ acting on nearest-neighbor qubits with periodic boundary conditions, with each gate parameterized by a different angle. The total number of trainable parameters therefore scales as $M = 4 n L$, and in the numerics we take $L = n$.

\begin{figure*}
    \centering    
    \includegraphics[width=1\linewidth]{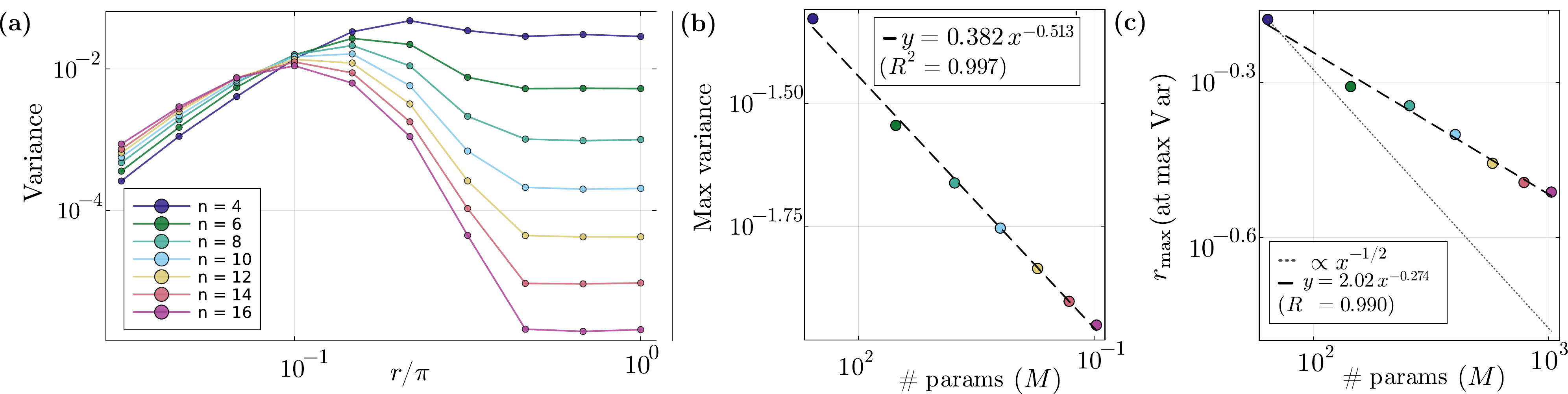}
    \caption{\textbf{Variance of the loss function in a hypercube of side $2r$ in parameter space}
(i.e., side $2r/\pi$ in the rescaled $r/\pi$ units used on the horizontal axis).
For each system size $n$ and radius $r$, we sample random directions in the
variational parameter space around the optimized point $\theta^\star$ obtained
from warm-start training, and estimate the variance of the loss at $x_2 = 0.2$, with the Heisenberg model described in Eq.~\eqref{eq:Ising_Ham}.
\textbf{(a)} Variance as a function of $r/\pi$ for different qubit numbers $n$
(log–log scale). \textbf{(b)} Maximum variance  as a function of the total number of variational parameters $M$. \textbf{(c)} Radius
$r_{\max}$ at which the variance is maximal, plotted versus $M$. In panel \textbf{(c)} we also include a reference line proportional to $M^{-1/2}$, corresponding
to regime $r \lesssim 1/\sqrt{M}$ for which Theorem~\ref{thm:main} guarantees at most polynomially small variance. In contrast, the
fitted behavior $r_{\max} \sim M^{-0.274}$ shows a more favorable empirical
scaling. All runs use $L = n$ layers, synthetic shot noise
with $\texttt{nshots}=10^4$ for the underlying optimization and $10^4$ samples to estimate the variance.
}
    \label{fig:variance_scaling}
\end{figure*}

In Fig.~\ref{fig:variance_scaling} $(a)$ we plot the variance of the loss $\LC_{\rm VQE}$ for $x_2 = 0.2$, corresponding to the second iteration step. We compute the variance for different values of $r$, the half-width of the hypercube centered at the solution $\thv_1^*$, from which we sample the new initial parameters in random directions. As expected, for $r =\pi$ the variance vanishes exponentially with the number of qubits. This is because in this regime, the initialization is uniform over the full landscape, so the concentration of the loss function phenomena (BP) manifests. However, as the size of the hypercube shrinks, the variance increases and its decay transitions to a polynomial scaling with the number of qubits.

We recall that these are sufficient conditions: they guarantee substantial loss variation and gradients when they are satisfied, but they do not imply that training must fail when they are violated. In particular, the bound ensures that under these conditions, the initialization strategy is guaranteed to be in a region with non-negligible gradients. A natural follow-up question is whether one can still train from larger neighborhoods than those covered by the theorem.

We address this question numerically in Fig.~\ref{fig:variance_scaling} (b) and (c). There, we show the relationship between the maximum value of the variance and the value of $r$ that maximizes it, with respect to the number of variational parameters $M$ (i.e., the dimension of $\thv$). We find that $r_{\max}$ scales roughly as $M^{-1/3}$, very close to the $M^{-1/2}$ bound predicted by Theorem~\ref{thm:main}. In fact, the observed scaling is slightly more favorable than the worst-case bound provided by Theorem~\ref{thm:main}, but almost saturates the bound.

\subsection{Training via iterative methods}\label{sec:training}

So far we have shown that iterative methods can work as a warm starting strategy, ensuring that the variational optimization can be initialized in a locally trainable region. However, this is not sufficient to guarantee that training will converge to a sufficiently good minimum. In this section we train different variational algorithms to assess practical advantages and limitations of iterative training methods. 

We start by presenting the models that we will study using the VQE and Meta-VQE algorithms. In particular, we consider one-dimensional spin chains governed by two different Hamiltonians: the anisotropic XY model and the transverse-field Ising model. Both are parameterized by a control parameter ${x}$ that interpolates between distinct interaction regimes. 
The anisotropic XY model is defined as
\begin{equation}
    H_{\mathrm{XY}}(x)
    = -J \big[H_{\sx} + H_{\sy}+x(H_{\sx} - H_{\sy})\big],
    \label{eq:XY_Ham}
\end{equation}
where $H_{\sigma_{j}} = \sum_{i=1}^{n-1}\sigma_j^{(i)} \sigma_j^{(i+1)},\,\sigma_j\in\{\sx,\sy\} $ and
where $x$ controls the anisotropy between $\sx\sx$ and $\sy\sy$ couplings.  
For $x=0$ the model reduces to the isotropic XX chain, while for $x=1$ it reproduces the Ising limit.
The generalized Ising Hamiltonian used in our simulations is 
 \begin{align}
     H_{\mathrm{I}}(x)
     = &-J \sum_{i=1}^{n} Z_i Z_{i+1}
       - x \sum_{i=1}^{n} X_i\\
       &- \, Y_1 X_2 \cdots X_{n-1} Y_n,
\end{align}
 where $J$ sets the strength of the nearest-neighbor coupling and $x$ acts as a tunable transverse field.  The last term $X_i
       - \, Y_1 X_2 \cdots X_{n-1} Y_n$ cancels the periodic boundary term $Z_1 Z_n$ after the Jordan-Wigner
transformation in order to solve the system as it was infinite~\cite{cervera2018exact}.

\begin{figure*}[ht!]
    \centering    \includegraphics[width=1\linewidth]{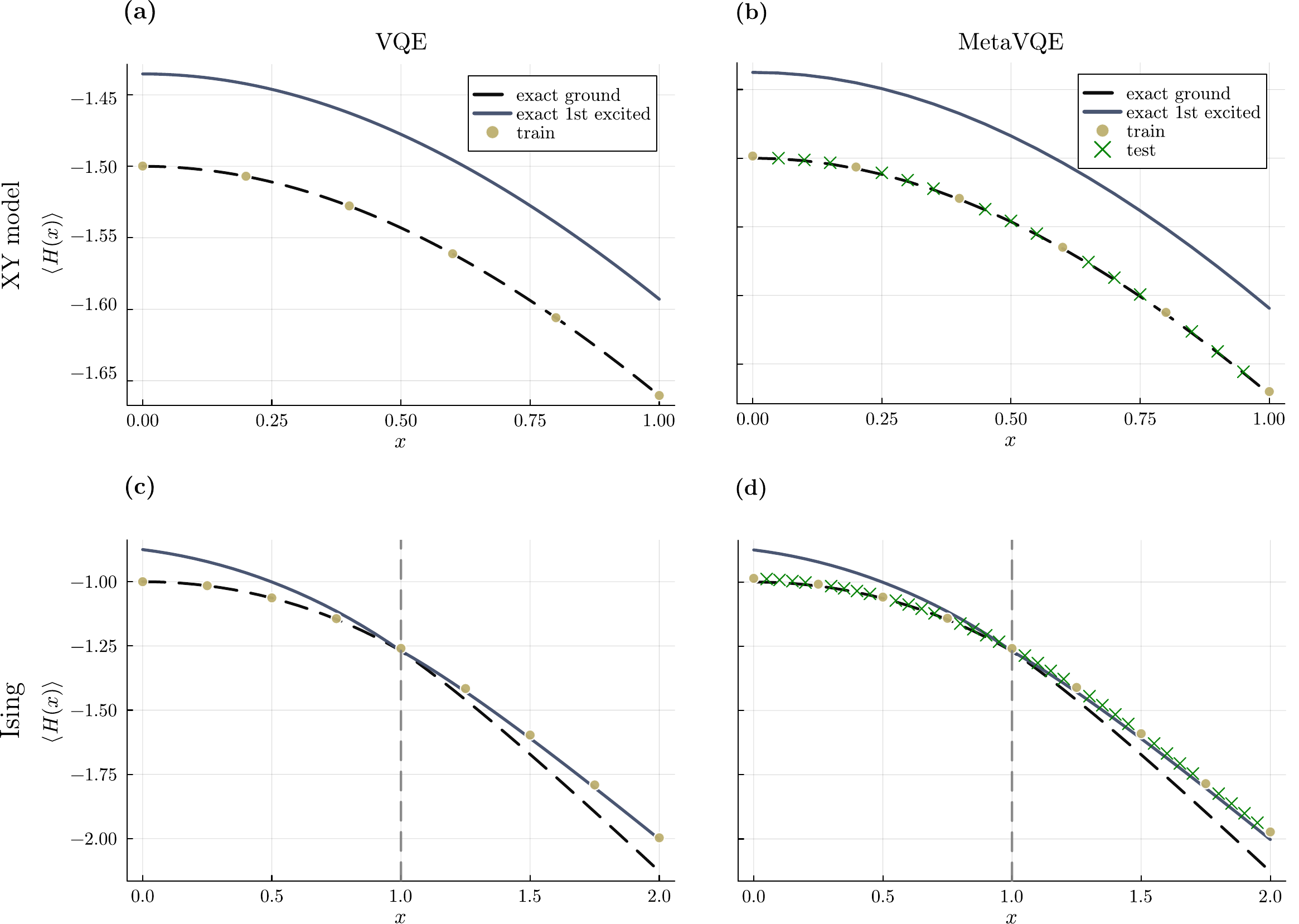}
   \caption{\textbf{Energy spectra learned with warm-start VQE and warm-start Meta-VQE.}
Rows correspond to two target models (top: XY; bottom: transverse-field Ising), while columns compare VQE (left) and Meta-VQE (right). In all panels we use \emph{warm starts}, i.e., for each value of the field parameter $x$ the optimization is initialized with the parameters obtained at the previous point.
Solid lines show the exact ground-state and first-excited energies as functions of $x$, and markers show the corresponding estimates obtained under synthetic shot noise with $\texttt{nshots}=10^4$ (circles: training points; crosses: test points).
In the Ising row, the vertical dashed line indicates the critical point at $x=1$.
All data are for $n=16$ spins and circuit depth $L=16$.}
    \label{fig:training}
\end{figure*}

In order to assess whether training remains feasible in a realistic setting, it is important to account for finite-shot measurement noise. To this end, we model measurement uncertainty through a synthetic shot-noise model at the level of measurement groups. For a fixed total shot budget $n_{\mathrm{shots}}$, we partition the Hamiltonian into $G$ disjoint groups of Pauli terms, $H(x)=\sum_{g=1}^G H_g(x)$, where each group corresponds to a single measurement setting and is assigned a group-dependent shot budget $S_g$ with $\sum_{g=1}^G S_g=n_{\mathrm{shots}}$. For a given parameter vector $\boldsymbol{\theta}$, we then define the (synthetic) finite-shot energy estimator as
\begin{equation}
    \widehat{E}(\boldsymbol{\theta})
    :=
    \sum_{g=1}^G
    \left(
        \LC_g(\thv)
        +
        \sqrt {V_g(\boldsymbol{\theta})}\,\xi_g
    \right),
    \quad
    \xi_g\sim\mathcal{N}(0,1),
    \label{eq:synthetic_shot_energy_groups}
\end{equation}
where $ \LC_g(\thv)$ is the loss function associated to the exact contribution of group $g$, and $V_g(\boldsymbol{\theta})$ is chosen to match the leading-order finite-shot variance of estimating that group with $S_g$ shots (see Appendix~\ref{ap:shot_noise}). Concretely, under the standard approximation that intra-group covariances are negligible,
\begin{equation}
    V_g(\boldsymbol{\theta})
    \approx
    \sum_{\alpha\in\mathcal{I}_g} c_\alpha^2\,
    \frac{1-\langle P_\alpha\rangle_{\boldsymbol{\theta}}^{\,2}}{S_\alpha},
    \qquad
    \sum_{\alpha\in\mathcal{I}_g} S_\alpha = S_g,
    \label{eq:sigma_group_maintext}
\end{equation}
with $S_\alpha$ the shots allocated to each Pauli term within the group. We draw independent noise $\{\xi_g\}_{g=1}^G$ for each circuit evaluation, reflecting independent measurement batches across settings. This construction reproduces the expected statistical variance of finite-shot experiments.

In Fig.~\ref{fig:training} $(a)$ and $(b)$ we show an iterative training method for both VQE and Meta-VQE applied to the $H_{\rm XY}$ Hamiltonian in Eq.~\eqref{eq:XY_Ham}. In this plot we can clearly see that the iterative method works. Indeed for both strategies the training can follow the ground state. Not only that, but we also see that for Meta-VQE the test points also match the exact energies. 

The situation changes for the Ising Hamiltonian $H_I$ in Eq.~\eqref{eq:Ising_Ham} (Fig.~\ref{fig:training} panels (c) and (d)). Near the level crossing at $x=1$, the optimization fails to remain on the true ground state. Instead, it continues along the smooth continuation of the previously learned state, in agreement with the behavior predicted by Theorem~\ref{thm:main}. This highlights how level crossings and gap closings can obstruct iterative training and, in some cases, prevent convergence to the global optimum.

This behavior can be understood from the locality of the optimization. Before the crossing, the tracked state coincides with the global minimum; after the crossing, the same continuous trajectory persists but corresponds only to a \emph{local} minimum. Since the optimizer explores the landscape locally, nearby parameter updates typically increase the energy, making it difficult to transition to a different branch that becomes energetically favorable after the crossing. Escaping this would require an initialization already aligned with the post-crossing ground state, which cannot be guaranteed without prior knowledge of the solution.

Consequently, the issue at hand is not due to barren plateau phenomena or shot noise, but rather a fundamental limitation of local, iterative optimization strategy in the presence of spectral crossings. In this sense, the phenomenon is closely related to the bottlenecks encountered in adiabatic (or adiabatic-like) ground-state preparation strategies when the gap closes.

The existence of gradient paths between the old and new global minima, previously dubbed fertile valleys~\cite{puig2024variational}, has attracted attention recently. While 
as ground state finding for many-body Hamiltonians is QMA-complete~\cite{kempe2004complexity}, fertile valleys provably cannot exist in the worst case, numerical evidence suggests that they may sometimes be possible to find for some simpler problems. How widespread or rare this phenomenon is, 
as well as whether there are particular strategies that can make these fertile valleys more likely, is thus one important open direction for future research.

\section{Discussion}

This work investigates iterative warm-start training for variational optimization, where each instance is initialized using parameters optimized for a neighboring instance. By viewing VQE and Meta-VQE through the lens of a discretized adiabatic deformation, we analyzed the conditions under which  the optimization remains within an informative region of the parameter space. Theorem~\ref{thm:main} establishes a lower bound on the loss variance within a hypercube centered at the previous optimum, provided that: (i) consecutive Hamiltonians are sufficiently close, (ii) the previous step achieved adequate accuracy, and (iii) the initial gate induces a nontrivial local state variation. We find that the permissible warm-start neighborhood shrinks as the number of variational parameters increases, but expands with a larger spectral gap. Within this neighborhood, the loss landscape retains sufficient local variation to provide useful gradients, thereby avoiding barren plateaus. These analytical findings are corroborated by numerical landscape probes for system sizes up to $n=16$. As the warm-start neighborhood contracts from ``global'' initializations (large $r$) toward a local regime, the variance transitions from an exponentially vanishing scaling to a polynomial one. 

Training simulations under a synthetic shot-noise model further demonstrate that warm-start VQE and Meta-VQE reliably track the ground-state energy along paths with a sufficient spectral gap. However, this behavior changes qualitatively when the gap closes: near a level crossing, the optimization tends to follow the smooth continuation of the pre-crossing solution, which becomes a \emph{local} minimum after the crossing. This failure mode is consistent with the gap sensitivity identified in Theorem~\ref{thm:main} and reflects the limitations of local iterative optimization in the presence of spectral crossings.

Despite these limitations, this iterative framework offers distinct advantages. Unlike standard adiabatic evolution, which often requires significant evolution time (or circuit depth, if the implementation is digitalized) to satisfy the adiabatic theorem, this variational approach can search for a ``shortcut to adiabaticity'' within a constrained, lower-depth ansatz. Furthermore, the ability to reliably produce approximate ground states makes this method highly valuable for fault-tolerant algorithms; specifically, these states can serve as high-overlap inputs, thereby reducing the total runtime required to achieve high-precision energy estimates. 

It is perhaps worth highlighting that recent results have proven that local regions of variational landscapes may be modeled by a surrogate~\cite{lerch2024efficient}. In this approach, one can take several measurements at the beginning of the protocol to simulate regions of the quantum landscape (and train using those regions), which is potentially applied to the method proposed here. While this typically requires a comparable overall amount of quantum hardware access, it can offer practical advantages in real deployments—most notably by reducing the need for tight quantum-classical feedback loops, making the optimization less sensitive to time-dependent noise drift and calibration changes, and enabling more efficient workflows when access is mediated through cloud queues and intermittent device availability.

Several open directions follow naturally from our work here. To mitigate the issue of level crossings, one might track not only the ground state but also a subspace of low-lying excited states, ensuring the target branch can be recovered after eigenvalue reordering. Variational Quantum Deflation~\cite{higgott2019variational} offers a concrete framework for this, though it remains an open question whether multi-state tracking provides provable guarantees in the iterative setting. 

Another strategy might be to exploit the control space: by starting from various solvable Hamiltonians, one could attempt to construct trajectories that circumvent spectral gap closures entirely. This enables the variational discovery of shortcuts to adiabaticity through the careful selection of the Hamiltonian path. Ultimately, this suggests a paradigm complementary to the heuristic parameter search of non-iterative versions of VQE: instead of simply guessing initialization points and ansatz, one shapes the optimization path itself. By enforcing symmetries or conserved quantities along the Hamiltonian trajectory, the optimizer can be confined to a protected subspace, shielding the process from symmetry-breaking level crossings.

Taken together, these results position iterative warm-start variational methods as a potential route to overcoming trainability barriers in variational quantum algorithms. Furthermore, by exposing their deep connection to adiabatic evolution, point toward a new paradigm in which the mature toolbox of adiabatic control and shortcut techniques can be systematically repurposed. Indeed, this opens the door to designing scalable and robust variational training strategies.

\newpage

\acknowledgments

We thank Marc Drudis, Paolo Braccia, Santiago Llorens, and Pauline Besserve for insightful discussions. 
R.P. acknowledges the support of the SNF Quantum Flagship Replacement Scheme (Grant No. 215933). B. C. acknowledges funding from the Spanish Ministry for Digital Transformation and the Civil Service of the Spanish Government through the QUANTUM ENIA project call - Quantum Spain, EU, through the Recovery, Transformation and Resilience Plan – NextGenerationEU, within the framework of Digital Spain 2026. A.C.-L. acknowledges funding by the European Union, supported by the EuroHPC Joint Undertaking and its members under the Grant Agreement Nº 101159808, including top-up funding by Ministry for Digital Transformation and the Civil Service of the Spanish Government, and from the grant RYC2022-037769-I funded by MICIU/ AEI/ 10.13039/ 501100011033 and by “ESF+". Z.H. acknowledges support from the Sandoz Family Foundation-Monique de Meuron program for Academic Promotion. A.P.S acknowledges funding from Swiss National Science Foundation through the grant TMPFP2\_234085, and Juan Carrasquilla for his support. 

\bibliography{quantum,bib}

\onecolumngrid

\vfill
\newpage

\appendix

\section{Preliminaries}\label{app:preliminaries}

\subsection{Taylor remainder theorem}\label{app:taylor}

We present the Taylor remainder theorem which expresses a multivariate differentiable function as a series expansion. We refer the reader to Ref.~\cite{taylorbook} for further details.    

\begin{theorem}[Taylor reminder theorem] \label{thm:taylor}
Consider a function $f(\vec{x})$ such that $f: \mathbb{R}^N \rightarrow \mathbb{R}$ which is $K$ times differentiable around a point $\vec a \in \mathbb R^N$, for $K$ a positive integer $K$. The function $f(\vec{x})$ can be expanded around $\vec{a}$ as
\begin{align}
    f(\vec{x}) = \sum_{k = 0}^{K} \sum_{i_1, i_2, ...,i_k}^N \frac{1}{k!} \left( \frac{\partial^k  f(\vec{x})}{\partial x_{i_1} \partial x_{i_2} ... \partial x_{i_k}} \right)\bigg|_{\vec{x} = \vec{a}} (x_{i_1} - a_{i_1})(x_{i_2} - a_{i_2}) ... (x_{i_k} - a_{i_k}) + R_{K,\vec{a}}(\vec{x}) \;,
\end{align}
where the remainder is of the form
\begin{align}
     R_{K,\vec{a}}(\vec{x}) = \sum_{i_1, i_2, ...,i_{K+1}}^N \frac{1}{(K+1)!} \left( \frac{\partial^{K+1}  f(\vec{x})}{\partial x_{i_1} \partial x_{i_2} ... \partial x_{i_{K+1}}} \right)\bigg|_{\vec{x} = \vec{\nu}} (x_{i_1} - a_{i_1})(x_{i_2} - a_{i_2}) ... (x_{i_{K+1}} - a_{i_{K+1}}) \;,
\end{align}
with $\vec{\nu} = c \vec{x} + (1-c) \vec{a}$ for some $c \in [0,1]$. 
\end{theorem}

\subsection{Variance decomposition of a multivariable function}\label{app:variance_decomp}

Here we show how the variance of a multivariate function can be decomposed into a sum of expected values and variances for different variables.

\begin{proposition}\label{prop:var_decomp}
Consider a multi-variable function $f(\thv)=f(\th_1,\dots,\th_M)$ depending on $M$ parameters such that $f(\thv): \mathbb{R}^M \rightarrow \mathbb{R}$. We assume that each parameter is i.i.d  sampled from some distribution $\DC$ i.e., $\thv \sim \DC^{\otimes M}$. Then we can express the variance of $f(\thv)$ as
\begin{align}
    \Var_{\vec{\th} \sim \DC^{\otimes M}}\left[f(\vec{\th})\right] = \sum_{k=1}^M \Ebb_{\bm\th_{\overline{k + 1}}}[\Var_{\th_k}[\Ebb_{\bm\th_{\underline{k-1}}}[f(\vec{\th})]]]\,,
\end{align}
and therefore we can lower bound the variance as
\begin{align}
     \Var_{\vec{\th} \sim \DC^{\otimes M}}\left[f(\vec{\th})\right]\geq \Var_{\th_M}[\Ebb_{{\thv_{\overline{M}}}}[f(\thv)]]\,,
\end{align}
where $\thv_{\underline{k}} = (\th_1,..,\th_{k})$ and $\thv_{\overline{k}} = (\th_{k + 1}, \ldots, \th_M)$.
\end{proposition}
\begin{proof}

     The proof can be obtained by recursion over the number of independent parameters $M$. Precisely, let us introduce the function $f_k$ defined recursively as follows
    \begin{align}
        f_{\overline{0}}(\vec\theta)&= f(\vec{\th})\\
        f_{\overline{k}}(\thv) &= \Ebb_{\theta_{k}}[f_{\overline{k-1}}(\thv)] = \Ebb_{\thv_{\underline{k}}}[f(\thv)]
    \end{align}
    Here, one can see that the function $f_{\overline{k}}$ only depends on the parameters $\vec{\th}_{\overline{k}}$. 
    Moreover, we can show that $\forall \; 0\leq  k \leq m$ 
    \begin{align}
        \Var_{\thv_{\overline{k}}}[f_{\overline{k}}(\thv_{\overline{k}})] &= \Ebb_{\thv_{\overline{k}}}[f_{\overline{k}}^2(\thv_{\overline{k}})] - \Ebb_{\thv_{\overline{k}}}[f_{\overline{k}}(\thv_{\overline{k}})]^2 \label{eq:recursive_f_start}\\
        &=  \Ebb_{\thv_{\overline{k+1}}}[\Ebb_{\thv_{k + 1}}[f_{\overline{k}}^2(\thv_{\overline{k}})]] - \Ebb_{\thv_{\overline{k+1}}}[\Ebb_{\thv_{k + 1}}[f_{\overline{k}}(\thv_{\overline{k}})]]^2\\
        & \begin{aligned}= \Ebb_{\thv_{\overline{k+1}}}&[\Ebb_{\thv_{k + 1}}[f_{\overline{k}}^2(\thv_{\overline{k}})]]- \Ebb_{\thv_{\overline{k+1}}}[\Ebb_{\thv_{k + 1}}[f_{\overline{k}}(\thv_{\overline{k}})]^2] \\ & + \Ebb_{\thv_{\overline{k+1}}}[\Ebb_{\thv_{k + 1}}[f_{\overline{k}}(\thv_{\overline{k}})]^2] - \Ebb_{\thv_{\overline{k+1}}}[\Ebb_{\thv_{k + 1}}[f_{\overline{k}}(\thv_{\overline{k}})]]^2\end{aligned}\\
        &= \Ebb_{\thv_{\overline{k+1}}}[\Var_{\thv_{k + 1}}[f_{\overline{k}}(\thv_{\overline{k}})]] + \Var_{\thv_{\overline{k+1}}}[\Ebb_{\thv_{k + 1}}[f_{\overline{k}}(\thv_{\overline{k}})]]\\
        &= \Ebb_{\thv_{\overline{k+1}}}[\Var_{\thv_{k + 1}}[f_{\overline{k}}(\thv_{\overline{k}})]] + \Var_{\thv_{\overline{k+1}}}[f_{\overline{k+1}}(\thv_{\overline{k+1}})]
    \end{align}

    Hence, we have that
    \begin{align}
    \Var[f(\vec{\th})] 
    = \Var_{\thv_{\overline{0}}}[f_{\overline 0}(\thv_{\overline{0}})] & =\Ebb_{\thv_{\overline{1}}}[\Var_{\thv_{1}}[f_{\overline 0}(\thv_{\overline{0}})]] + \Var_{\thv_{\overline{1}}}[f_{\overline 1}(\thv_{\overline{1}})]\\
    & =\Ebb_{\thv_{\overline{1}}}[\Var_{\thv_{1}}[f_{\overline 0}(\thv_{\overline{0}})]] + \Ebb_{\thv_{\overline{2}}}[\Var_{\thv_{2}}[f_{\overline 2}(\thv_{\overline{2}})]] + \Var_{\thv_{\overline{2}}}[f_{\overline{2}}(\thv_{\overline{2}})] \\ 
    & = \sum_{k=0}^{m-1} \Ebb_{\thv_{\overline{k + 1}}}[\Var_{\thv_{k + 1}}[f_{\overline{k}}(\thv_{\overline{k}})]]\\
    &= \sum_{k=1}^{m} \Ebb_{\thv_{\overline{k}}}[\Var_{\thv_{k - 1}}[\Ebb_{\thv_{\underline{k - 1}}}[f(\vec{\th})]]] \label{eq:recursive_f_end}
    \end{align}

    Furthermore, because the variance of a function is always positive, every term in the previous sum in Eq.~\eqref{eq:recursive_f_end} is larger or equal to zero. Thus, we can use that the sum of positive terms is lower bounded by one of the terms in the sum to find the following lowerbound
    \begin{align}
         \Var[f(\vec{\th})] \geq \Var_{\th_1}[\Ebb_{{\thv_{\overline{1}}}}[f(\thv)]]
    \end{align}
    which concludes the proof. Additionally, any permutation of the parameters allows to find the same relationship for any individual $0\leq k \leq M$.
\end{proof}

\section{Lower bound on the variance of the VQE loss function}\label{app:noencoding}
In this section, we analyze the variational problem in the absence of explicit data encoding within the quantum circuit. 
That is, the parameters $x_j$ appears only in the Hamiltonian, while the circuit parameters $\thv$ are optimized independently for each value of $x_j$. 
We show that if the parameters $\thv^*$ prepare the ground state of $H(x_{j-1})$ with precision $\epsilon$, and the step 
$\Delta x_j = x_j - x_{j-1}$ is sufficiently small, then the variance remains non-vanishing, ensuring that we could estimate the gradients that lead to a solution $\boldsymbol{\theta}$ for $H(x_j)$. 

This analysis constitutes the first step toward generalizing the result to the case where the loss function is a sum of $N$ terms evaluated with a parameterized quantum circuit that includes explicit encodings $x_j$.  The result is given in the following theorem:

\begin{theorem}
     Let the loss function associated with the $j$-th Hamiltonian be 
     \begin{equation}
         \LC_j(\thv)= \expval{U^\dagger (\thv) H(\paramh_j) U(\thv)}{\psi}.
     \end{equation}
     Assume that the parameters $\thv^*$ generate a state 
    $U(\boldsymbol{\theta}^*)|\psi\rangle$ whose fidelity with the ground state 
    $|{\rm gs}_{j-1}\rangle$ of $H(x_{j-1})$ satisfies
    \begin{equation}
        |\langle {\rm gs}_{j-1} | U(\boldsymbol{\theta}^*)|\psi\rangle|^2 = 1 - \epsilon^2,
    \end{equation}
    i.e., the prepared state is $\epsilon$-close to the true ground state 
    in the Hilbert-space $2$-norm,  in particular we consider $\epsilon\leq 1/\sqrt 2 $. Assume moreover that the first circuit generator $P_1$ acts non-trivially on the input state, in the sense that $|\expval{P_1}{\psi}|^2 = 0$.
    
    Suppose further that the Hamiltonians at consecutive steps is given by
    \[
        H(x_j) = H(x_{j-1}) + \Delta x_j\, H_1,
    \]
    where $H_1$ is a fixed Hermitian operator, and that the step size obeys
    \begin{equation}
        |\Delta x_j|
        \leq
        \widetilde{\gamma}\,
        \frac{(1-2\epsilon^2)\,|\Delta_{\mathrm{gap}}|}{\|H_1\|_s},
        \label{eq:delta_x_condition}
    \end{equation}
    with $\widetilde{\gamma}\in(0,1)$, and $|\Delta_{\mathrm{gap}}|$ being the absolute value of the gap of $H(x_j)$.
    Then, for a uniform distribution 
    $\mathcal{D}(\boldsymbol{\theta}^*, r)$ centered at $\boldsymbol{\theta}^*$ 
    within a hypercube of side $2r$, the variance of 
    $\mathcal{L}_j(\boldsymbol{\theta})$ satisfies
    \begin{equation}
        \Var_{\thv\sim\DC(\thv^*, r)} [\LC_j(\thv)] \geq \left(  1-\frac{4r^2}{7} \right) \frac{4 r^4}{45} \left( (1-\gamma)(1-\widetilde{\gamma})(1-2\epsilon^2)|\Delta_{\rm gap}| \right)^2,
    \end{equation}
    given that 
    \begin{equation}
        r^2<\gamma\frac{3}{(M-1)}\frac{(1-2\epsilon^2)(1-\widetilde{\gamma})|\Delta_{\rm gap}|}{\|H(\paramh_j)\|_s + (1-2\epsilon^2)(1-\widetilde{\gamma})|\Delta_{\rm gap}|},
    \end{equation}
    Here $\gamma,\widetilde{\gamma}\in(0,1)$ are constants, $\epsilon$ quantifies the infidelity with the previous ground state, $|\Delta_{\mathrm{gap}}|$ denotes the spectral gap of $H(x_j)$,
    $\|\cdot\|_s$ is the spectral semi-norm, and $M$ is the number of trainable parameters in the circuit.

    Consequently, if $r \in \Theta(\mathrm{poly}(n))$, then $\Var_{\boldsymbol{\alpha}\sim\mathcal{D}(\vec{0},r)}[\mathcal{L}_j(\boldsymbol{\alpha})]
    \in \Omega(\mathrm{poly}(n))$.
\end{theorem}
\begin{proof}

First we recall that the loss function is defined as
    \begin{equation}\label{eq:loss_app_exact_form_simpler}
        \LC_j(\thv) = \expval{U^\dagger (\thv) H(\paramh_j) U(\thv)}{\psi}\,, 
    \end{equation}
    where $\ket{\psi}$ represents the initial state, $H(\paramh_j)$ a parametrized observable and $U(\thv)$ a parametrized quantum circuit. Finally, $x_j$ is a parameter of the Hamiltonian.
    
    The unitary matrices are defined as 
    \begin{align}
        U(\thv) = \prod_{i=1}^M  U_i(\th_i)V_i\,,
    \end{align}
    where $\thv$ is a vector with the variational parameters $\thv = (\th_1, \th_2,...,\th_M)$ and each unitary gate takes the form of $U_k(\th_k) = e^{-i \th_k P_k}$, where $P_k$ is a Pauli string, and $V_k$ are arbitrary non-parameterized unitary transformations. 

    The loss function $\LC_N(\thv)$  is trained for a new $\paramh_N$. Note that at each training iteration we change the loss function. Indeed, we have a different loss function for every $\paramh_1,...\paramh_N$, however, as we will see, we need the different Hamiltonians to be sufficiently close to each other. We assume training as an iterative method. In each new iteration, we will initialize the parameters around the optimal solution of the previous iteration $\thv^*_{j-1}$. Crucially, the parametrized gates can be re-expressed as perturbations $\alv$ around the previous optimal solution, i.e. $\th_i = \th_i^*+\al_i$ as shown below
    \begin{align}
        U(\thv) =& \prod_{i=1}^M V_i U_i(\th_i) \\
        = & \prod_{i=1}^M V_i U_i(\th_i^*)U_i(\al_i) \\
        =&  \prod_{i=1}^M \widetilde{V}_i(\th_i^*)U_i(\al_i)\, , 
    \end{align}
    where we used that $e^{-i (a+b)P} =e^{-i aP}e^{-i bP} $ in the first equality to separate the parametrized unitary into two. In the final equality we absorbed the $\theta^*$-dependent $U_i$ into $V_i$, i.e. $\widetilde{V}_i(\th_i^*) = V_iU_i(\th_i^*) $. To ease the notation, we will denote $\widetilde{V}_i(\th_i^*)$ as just $\widetilde{V_i}$.

    We consider a region of parameters around the previous solution of the form
    \begin{align}
        \VC(\thv^*,r) = \{\thv = \thv^* + \alv\,|\,\alv\in[-r,r]^M.\}
    \end{align}
    where $r$ is the characteristic length of the region. 

    We are interested in the variance of the loss function in Eq.~\eqref{eq:loss_app_exact_form_simpler} over the region $\VC(\thv^*,r)$ such that the parameters are uniformly sampled 
    \begin{align}
        \DC(\thv^*,r) := {\rm Unif}[ \VC(\thv^*,r)]\,,
    \end{align}
    or in other words we are interested in
    \begin{align}\label{eq:variance_of_loss_app_exact_simpler}
     \Var_{\thv\sim\DC(\thv^*, r)} [\LC_j(\thv)] = &\Var_{\alv\sim\DC(\vec{0}, r)} [\LC_j(\thv=\thv^* + \alv)]\,.
     \end{align}

     Using this, we now move to finding a general lower bound for the variance of the loss function.

    From now onward we are going to use $\alv$ as the variables of interest instead of $\thv$, as explained in the previous change of variable. Indeed, $\alv$ allows us to assume, without loss of generality, that the hypercube of interest is centered at zero.

 As we are interested in independently sampled parameters from the same distribution, we can use Proposition~\ref{prop:var_decomp} to immediately find a lower bound on Eq.~\eqref{eq:variance_of_loss_app_exact_simpler}
\begin{align}\label{eq:lowerbound_variance_simpler}
    \Var_{\alv\sim\DC(\vec{0},r)}[\LC(\alv)]\geq \Var_{\al_1\sim\DC(0,r)}[\Ebb_{\alv_{\overline{1}}\sim\DC(\vec{0}_{\overline{1}},r)}[\LC(\alv)]]\,
\end{align}
where we denote $\alv_{\overline{k-1}} = (\al_k,...\al_M)$, and similarly, $\vec{0} = (0,...,0)$ is the vector of $M-1$ null components.

We start by computing the expected values of Eq.~\eqref{eq:lowerbound_variance_simpler}. Because the variables are independent, we can compute the different expected values one-by-one, starting with $\al_2$ to $\al_{M}$, that is
\begin{align}
\Ebb_{\alv_{\overline{1}}\sim\DC(\vec{0}_{\overline{1}},r)}[\LC(\alv)] = & \Ebb_{\al_2}[... \Ebb_{\al_M}[\LC(\alv)]...]\,.
\label{eq:recursive_expectation_value}
\end{align}

To do this, we express each parametrized unitary $U_i(\al_i) = e^{-i \al_i P_i} = \cos(\al_i)\1  -i\sin(\al_i)P_i$. 

We start by evaluating the expectation value of the $M$-th variable, the first one that appears in Eq.~\eqref{eq:recursive_expectation_value}.
Therefore, we have 
\begin{align}
    \Ebb_{\al_M}[\LC(\alv)] =& \Ebb_{\al_M}[\expval{U^\dagger (\thv) H(\paramh_j) U(\thv)}{\psi} ] \\ 
    =&\expval{\left( \prod_{i=1}^{M-1} \widetilde{V}_i(\th_i^*)U_i(\al_i)\right)^\dagger \underbrace{\Ebb_{\al_M}[U_M^\dagger(\al_M)\widetilde{V}_M^\dagger H(\paramh_j)\widetilde{V}_M U_M(\al_M)]}\left( \prod_{i=1}^{M-1} \widetilde{V}_i(\th_i^*)U_i(\al_i)\right)}{\psi} \,,
\end{align}

where in the last equation have make explicit the gates that depend on $\alpha_M$ and we have used linearity of expectation and of the inner product . Now, we focus only on the expectation value over the under-braced quantity.
\begin{align}
    \Ebb_{\al_M\sim\DC(0,r)}[ U_M^\dagger(\al_M)\widetilde{V}_M^\dagger H(\paramh_j)\widetilde{V}_M U_M(\al_M)] =& \Ebb_{\al_M\sim\DC(0,r)}[\cos^2(\al_M)\widetilde{V}_M^\dagger H(\paramh_j)\widetilde{V}_M   \\
    &+\sin^2(\al_M)P_M\widetilde{V}_M^\dagger H(\paramh_j)\widetilde{V}_M P_M \\
    &- i\sin(\al_M)\cos(\al_M)[P_M, \widetilde{V}_M^\dagger H\widetilde{V}_M] ]\label{eq:expansion_in_sin_cos_simpler}\\
    =&  k_{+} \widetilde{V}_M^\dagger H(\paramh_j)\widetilde{V}_M +  k_{-} P_M  \widetilde{V}_M ^\dagger H(\paramh_j)\widetilde{V}_M P_M\label{eq:expval_onegate_simpler}.
\end{align}
In the first equality we expanded the exponential into its trigonometric expansion. Then, in the last equation, we integrate explicitly, defining 
\begin{align}
    k_{+} :=& \frac{1}{2r}\int_{-r}^{r}d \al_i \cos^2(\al_i) = \frac{1}{2}(1+\sinc(2r))\label{eq:deffinition_kp_simpler}\\
    k_{-} :=& \frac{1}{2r}\int_{-r}^{r}d \al_i \sin^2(\al_i) = \frac{1}{2}(1-\sinc(2r)).
\end{align}
We also used the fact that the integral with an odd number of powers of $\sin$ over a symmetric space vanishes.

Because all the variables $\alpha_i$ are independent and follow the same distribution, we iteratively repeat the calculation in Eq.~\eqref{eq:expval_onegate_simpler} for all the subsequent parameters in a telescopic fashion. The only thing that changes is the central objects in Eq.~\eqref{eq:expval_onegate_simpler}, as we characterize now. 
Let us denote the object in Eq.~\eqref{eq:expval_onegate_simpler} as $H' := k_{+} \widetilde{V}_MH(\paramh_j)\widetilde{V}_M^\dagger +  k_{-} P_M  \widetilde{V}_M ^\dagger H(\paramh_j)\widetilde{V}_M P_M$. Now, integrating over the second variable $\alpha_{M-1}$ we have: 
\begin{align}
    &\Ebb_{\al_{M-1}\sim\DC(0,r)}\left[ \left(\prod_{s=1}^{M-1} \widetilde{V}_s U_s(\th_s) \right)^\dagger H'(\paramh_j)\left(\prod_{s=1}^{M-1} \widetilde{V}_s U_s(\th_s)\right) \right] \\
    &= \Ebb_{\al_{M-1}\sim\DC(0,r)}\left[ \left(\prod_{s=1}^{M-2} \widetilde{V}_sU_s(\th_s) \right)^\dagger U^\dagger_{M-1}(\alpha_{M-1}) \tilde V^\dagger_{M-1}H'(\paramh_j) \tilde V_{M-1} U_{M-1}(\alpha_{M-1})\left(\prod_{s=1}^{M-2} \widetilde{V}_s U_s(\th_s)\right) \right] \\
    & \hspace{4cm} = \left[ \left(\prod_{s=1}^{M-2} \widetilde{V}_sU_s(\th_s) \right)^\dagger H^{''}(x_j) \left(\prod_{s=1}^{M-2} \widetilde{V}_s U_s(\th_s)\right) \right] 
    \label{eq:integration_alpha_M-1}
\end{align}
where we denoted
\begin{align}
    H''(\paramh_j) =& k_+^2\widetilde{V}_{M-1}^\dagger\widetilde{V}_{M}^\dagger H(\paramh_j)\widetilde{V}_{M}\widetilde{V}_{M-1} + k_+k_-P_{M-1}\widetilde{V}_{M-1}^\dagger \widetilde{V}_{M}^\dagger H(\paramh_j)\widetilde{V}_{M} \widetilde{V}_{M-1} P_{M-1}  \\
    &+ k_-k_+\widetilde{V}^\dagger_{M-1}P_{M}\widetilde{V}^\dagger_{M}H(\paramh_j)\widetilde{V}_{M}P_{M}\widetilde{V}_{M-1} + k_-^2P_{M-1}\widetilde{V}_{M-1}^\dagger P_{M}\widetilde{V}_{M}^\dagger H(\paramh_j)\widetilde{V}_{M}P_{M} \widetilde{V}_{M-1}P_{M-1}.
\end{align}
In Eq.~\eqref{eq:integration_alpha_M-1} we have used the same derivations as in Eq.~\eqref{eq:expval_onegate_simpler}. We define the expression after integrating all the variables as
\begin{equation}
    H_{\rm eff}(x_j) := \Ebb_{\al_2\sim\DC(0,r)}\left[...\Ebb_{\al_{M-2}\sim\DC(0,r)} \left[ \left(\prod_{s=1}^{M-2} \widetilde{V}_sU_s(\th_s) \right)^\dagger H^{''}(x_j) \left(\prod_{s=1}^{M-2} \widetilde{V}_s U_s(\th_s)\right) \right] \cdots \right],
\end{equation}
which takes the explicit form
\begin{align}\label{eq:expectedval_hj_simpler}
    H_{\rm eff}(\paramh_j) = \widetilde{V}_1^\dagger \sum_{\vec{z}}\bigcirc_{l=2}^M k_{z_l}\mathcal{P}_{z_l}^{(l)} (H(\paramh_j)) \widetilde{V}_1
\end{align}
where $z_l\in\{+,-\}$, $\bigcirc$ represents the composition operation ($f\circ g = f(g)$), and the super-operator $\mathcal{P}_{z}^{(l)}(\cdot) $ is defined as 
\begin{align}
    \mathcal{P}_{z}^{(l)}(\cdot) =
    \begin{cases}
        \widetilde{V}_l^\dagger (\cdot) \widetilde{V}_l  & {\, \rm if \, } z\to +\\
       P_l \widetilde{V}_l^\dagger (\cdot)  \widetilde{V}_l P_l & {\, \rm if \, } z\to -
    \end{cases}
\end{align}
where it either applies the $l-$th non parametrized gates when the subscript $z$ is assigned to a plus ($+$) or conjugates the operator with two Pauli strings $P_l$ and then applies the non parametrized gates $V_l$ when $z$ is assigned to a minus ($-$). 

Because the expected value is linear with the sum (i.e. $\Ebb[a+b] = \Ebb[a] + \Ebb[b]$), we can now use Eq.~\eqref{eq:expectedval_hj_simpler} to compute the expected value of the loss function $\LC(\thv)$. Indeed, we can recover the explicit formula of the loss function in Eq.~\eqref{eq:loss_app_exact_form_simpler}, and use it to express the expected value in Eq.~\eqref{eq:lowerbound_variance_simpler}
\begin{align}
    \Ebb_{\alv_{\overline{1}}\sim\DC(\vec{0}_{\overline{1}},r)}[\LC(\thv)] =&  \expval{U_1^\dagger(\al_1)\widetilde{V}_1^\dagger H_{\rm eff}(\paramh_j)\widetilde{V}_1U_1(\al_1) }{\psi}.
\end{align}
Furthermore, we stress that the dependency on $\alpha_1$ is because we still need to compute the variance over it, which we address now. Using the previous equation, we are ready to compute the right hand side of Eq.~\eqref{eq:lowerbound_variance_simpler}, a lower bound on the variance over the constrained landscape. To do it, we start by expanding the gate $U_1$ again into its trigonometric form 
\begin{align}
     \Var_{\al_1\sim\DC(0,r)}[\Ebb_{\alv_{\overline{1}}\sim\DC(\vec{0}_{\overline{1}},r)}[\LC(\alv)] = &\Var_{\al_1\sim\DC(0,r)}\Bigg[\bra{\psi}\big(\cos^2(\th_1)  H_{\rm eff} (\paramh_j) +\sin^2(\th_1) P_1  H_{\rm eff} (\paramh_j)  P_1 \big)\ket{\psi}\Bigg]  \\
     &+ \Var_{\al_1\sim\DC(0,r)}\Bigg[ \sin(\th_1)\cos(\th_1) \expval{ i[P_M,  H(\paramh_j)] } {\psi} \Bigg]\label{eq:sincos_term_todrop_simpler}\\
     \geq &\Var_{\al_1\sim\DC(0,r)}\Bigg[ \bra{\psi}\bigg(\cos^2(\th_1) H_{\rm eff}(\paramh_j) +\sin^2(\th_1) P_1 H_{\rm eff}(\paramh_j) P_1\bigg) \ket{\psi}\Bigg] \label{eq:lowerbound_sincos_out_simpler}
\end{align}
where we used that $U_i(\al_i) = e^{-i \al_i P_i} = \cos(\al_i)\1  -i\sin(\al_i)P_i$. Furthermore, the term proportional to $\sin(\th_1)\cos(\th_1)$ is orthogonal to the terms $\sin^2(\th_1)$ and $\cos^2(\th_1)$. Therefore, the covariance between the two is zero.  Indeed, the integrals $\int\sin(\th_1)\cos^3(\th_1)$ or $\int\sin^3(\th_1)\cos(\th_1)$ over a symmetric space are zero, as it is an integral of an antisymmetric function times another symmetric function. Also, in the first inequality in Eq.~\eqref{eq:lowerbound_sincos_out_simpler}, we use the fact that the variance of a function is always positive. Therefore we can drop the variance in Eq.~\eqref{eq:sincos_term_todrop_simpler}, as it is positive.

Now we rewrite the variance as a sum of variances and covariances
\begin{align}
   &\Var_{\al_1\sim\DC(0,r)}[ \bra{\psi}\cos^2(\th_1)   H_{\rm eff}(\paramh_j)  +\sin^2(\th_1) P_1   H_{\rm eff}(\paramh_j)P_1  \ket{\psi}]\\
   =&\Var_{\al_1\sim\DC(0,r)}[\cos^2(\th_1)]|\expval{  H_{\rm eff}(\paramh_j) }{\psi}|^2 \\
    &+\Var_{\al_1\sim\DC(0,r)}[\sin^2(\th_1)]|\expval{  P_1 H_{\rm eff}(\paramh_j)P_1  }{\psi}|^2  \\
    &+ 2\Covar_{\al_1\sim\DC(0,r)}[\sin^2(\th_1),\cos^2(\th_1)]\expval{  P_1 H_{\rm eff}(\paramh_j)P_1  }{\psi}\expval{  H_{\rm eff}(\paramh_j) }{\psi}\label{eq:middle_calcul_for_lowerbound_simpler}
    \\
    =&h(r)\left( \expval{  H_{\rm eff}(\paramh_j) }{\psi} - \expval{  P_1 H_{\rm eff}(\paramh_j) P_1  }{\psi} \right)^2\label{eq:final_calcul_for_lowerbound_simpler}
\end{align}
where in the first equality we decomposed the variance of a sum into its covariances, i.e. $\Var[\varphi_1 + \varphi_2] = \Var[\varphi_1] + \Var[\varphi_2] + 2 \Covar[\varphi_1,\varphi_2]$. The covariance in defined as $\Covar_{\al\sim\DC}[\varphi_1(\al)\varphi_2(\al)] = \Ebb_{\al\sim\DC}[\varphi_1(\al)\varphi_2(\al)] - \Ebb_{\al\sim\DC}[\varphi_1(\al)]\Ebb_{\al\sim\DC}[\varphi_2(\al)]$. In the second equality we have introduced the function $h(r)$, given by 
\begin{align}
    h(r) = \frac{\cos (4 r)-1}{32 r^2} + 
  \frac{1+\sinc (4 r)}{8}\label{eq:h_function_simpler},
\end{align}
which conveniently corresponds to the three terms in Eq.~\eqref{eq:middle_calcul_for_lowerbound_simpler}, that is 
\begin{align}
   h(r) = 
       \Var_{\al_1\sim\DC(0,r)}[\sin^2(\th_1)] =
   \Var_{\al_1\sim\DC(0,r)}[\cos^2(\th_1)] =
    - \Covar_{\al_1\sim\DC(0,r)}[\cos^2(\th_1),\sin^2(\th_1)].
\label{eq:values_of_h_as_covars_simpler}
\end{align}

By recovering Eq.~\eqref{eq:lowerbound_variance_simpler}, and using the lower bound derived in Eq.~\eqref{eq:lowerbound_sincos_out_simpler} and the result of this term in Eq.~\eqref{eq:final_calcul_for_lowerbound_simpler}, we can finally derive a lower bound to the variance of the loss function that reads as follows
\begin{align}
    \Var_{\alv\sim\DC(\vec{0},r)}[\LC(\alv)]\geq &  h(r) (\expval{  \left (H_{\rm eff}(\paramh_j) - P_1 H_{\rm eff}(\paramh_j) P_1 \right) }{\psi})^2 \label{eq:lower_bound_variance_general_simpler}
\end{align}

Now we need to find a lower bound to the squared expectation value. To do so, we first show that the expecation value is negative, and thus we will find an upper bound of it (a lower bound on the term squared). To do so, we can check Eq.~\eqref{eq:expectedval_hj_simpler} to see the first term of $H_{\rm eff}(\paramh_j)$. By inspecting the equation, we see that $H_{\rm eff}(\paramh_j)$ is a convex combination of rotated Hamiltonians, as the parameters $k_{+} + k_{-} = 1$, and the difference in the Hamiltonian depends only on unitary matrices. Indeed
\begin{align}
    k_+ H_1 + k_-H_2 
\end{align}
is a convex combination of $H_1, H_2$. In the case of Eq.~\eqref{eq:expectedval_hj_simpler}, we see that each of the terms in the sum is the original $H(\paramh_j)$, but to which we apply a series of unitary transformations in the form of non-parametrized gates $\widetilde{V_i}$ and Pauli strings $P_l$.

Furthermore, $K_+:=  k_{+}^{M-1}$ is clearly the dominant term out of all of the other (for small enough values of $r$). Indeed, by Eq.~\eqref{eq:deffinition_kp_simpler}, we can use a Taylor expansion, which leads us to $k_{+}\sim 1 - \frac{r^2}{3} + \order{r^4}$, $k_{-}\sim \frac{r^2}{3}  + \order{r^4}$. Therefore, for small values of $r$ we expect the term proportional to $K_+$ to be dominant. In what follows we formalize this argument and find strict bounds on the value of $r$, such that the term $\expval{\widetilde{V}_1 \left (H_{\rm eff}(\paramh_j) - P_1 H_{\rm eff}(\paramh_j) P_1 \right)\widetilde{V}_1^\dagger}{\psi}$ is always negative. We start by rewriting this term as a sum of the terms depending on $K_+$: 
\begin{align}
    \expval{  \left (H_{\rm eff}(\paramh_j) - P_1 H_{\rm eff}(\paramh_j) P_1 \right) }{\psi} =& K_+ \expval{U(\thv^*) H(\paramh_j) U^\dagger (\thv^*) }{\psi}  \\
    &-  K_+ \expval{  P_1   U(\thv^*) H(\paramh_j) U^\dagger (\thv^*)   P_1   }{\psi}  \\
    &+\expval{   \widetilde{V}_1^\dagger \sum_{\vec{z}\neq \{+\}^{\otimes M}}\bigcirc_{l=2}^M k_{z_l}\mathcal{P}_{z_l}^{(l)} (H(\paramh_j)) \widetilde{V}_1}{\psi}  \\
    &- \expval{   P_1 \widetilde{V}_1^\dagger \sum_{\vec{z}\neq \{+\}^{\otimes M}}\bigcirc_{l=2}^M k_{z_l}\mathcal{P}_{z_l}^{(l)} (H(\paramh_j)) \widetilde{V}_1 P_1 }{\psi}\\
    =&  K_+\bra{\psi}[U(\thv^*) H U^\dagger (\thv^*) -  P_1   U(\thv^*) H(\paramh_j) U^\dagger (\thv^*)   P_1   ]\ket{\psi}  +\xi\label{eq:to_upperbound_xi_simpler}.
\end{align}
In the first equality we use the definition of $H_{\rm eff}$ introduced in Eq.~\eqref{eq:expectedval_hj_simpler}, while in the second equality we introduce $\xi$, a quantity that contains all the other terms. As discussed, $H_{\rm eff}(\paramh_j)$ can be understood as a convex combination of rotations of the Hamiltonian $H(\paramh_j)$. Therefore, $\xi$ is upper bounded by the norm of $H(\paramh_j)$: 
\begin{align}
    |\xi|=&\left|\sum_{\vec{z}\neq \{+\}^{\otimes M}} \expval{ \widetilde{V}_1 ^\dagger \bigcirc_{l=2}^M k_{z_l}\mathcal{P}_{z_l}^{(l)} (H(\paramh_j))\widetilde{V}_1}{\psi} - \expval{  P_1 \widetilde{V}_1^\dagger\bigcirc_{l=2}^M k_{z_l}\mathcal{P}_{z_l}^{(l)} (H(\paramh_j)) \widetilde{V}_1 P_1}{\psi}\right|\\
    \leq&\sum_{\vec{z}\neq \{+\}^{\otimes M}} \prod_{l=2}^M k_{z_l} \max_{\vec{z}'\in\{+,-\}^{\otimes M}}\| \widetilde{V}_1^\dagger\bigcirc_{k=2}^M \mathcal{P}_{z'_k}^{(k)} (H(\paramh_j))\widetilde{V}_1 -  P_1\widetilde{V}_1^\dagger \bigcirc_{k=2}^M \mathcal{P}_{z_k}^{(k)} (H(\paramh_j)) \widetilde{V}_1^\dagger P_1\|_{\infty}\\
    \leq & \|H(\paramh_j)\|_{s} \sum_{\vec{z}\neq \{+\}^{\otimes M}} \prod_{l=2}^M k_{z_l} = (1-K_+) \|H(\paramh_j)\|_{s}.
\end{align}
In the first equality we use a triangle inequality to split the expected value of the sum into the sum of absolute values. Then, we apply H\"older's inequality to maximize the value over $\psi$ that has 1-norm equal to one. In the last inequality, we use that the largest eigenvalue of the term inside the infinity norm can be at most the maximum difference between the eigenvalues of $H(\paramh_j)$, $\lambda_{\max} - \lambda_{\min}$. This is because the two terms are just rotations of the original Hamiltonian $H(\paramh_j)$. Note that we define the semi-norm $\|H\|_{s}:=\lambda_{\max} - \lambda_{\min}$ (see Ref.~\cite{boixo2007generalized, puig2024dynamical}). Using this result, we can upper bound equation Eq.~\eqref{eq:to_upperbound_xi_simpler}:
\begin{multline}
     \expval{ \left (H_{\rm eff}(\paramh_j) - P_1 H_{\rm eff}(\paramh_j) P_1 \right)}{\psi}\\  \leq  K_+\bra{\psi}[U^\dagger (\thv^*) H(\paramh_j) U (\thv^*)- P_1 U^\dagger(\thv^*) H(\paramh_j) U (\thv^*) P_1 ]\ket{\psi}  + (1-K_+)\|H(\paramh_j)\|_s.\label{eq:h0_h1_difference_simpler}
\end{multline}

Now, we focus on the remaining expected value in Eq.~\eqref{eq:h0_h1_difference_simpler} to show that this term is indeed negative. To do so, we recall that $\thv^*$ is the set of parameters we converged to in the previous training iteration. Without loss of generality we can write $U(\thv^*)\ket{\psi} = \sqrt{1-\epsilon^2}\ket{{\rm gs}_{j-1}} + \epsilon \ket{{\rm gs}_{j-1}^\perp}$, where $\ket{{\rm gs}_{j-1}}$ is the ground state of $H(\paramh_{j-1})$ and $\ket{{\rm gs}_{j-1}^\perp}$ is orthogonal to its ground state. We additionally assume that $|\expval{P_1}{\psi}|^2 = 0$, as it just means that the first generator of the circuit acts non-trivially on the initial input state. This also implies that $\ket{\psi},P_1\ket{\psi}$ are orthogonal, and therefore we can write $U (\thv^*,\paramh_j)P_1\ket{\psi} = \epsilon\ket{{\rm gs}_{j}} - \sqrt{1-\epsilon^2}\ket{{\rm gs}_{j}^\perp}$, where $\ket{{\rm gs}_{j}}$ is the ground state of $H(\paramh_j)$. Using this, we can readily simplify Eq.~\eqref{eq:h0_h1_difference_simpler} by rewriting $H(\paramh_j)= H(\paramh_{j-1}) +\Delta\paramh H_1$, where $H_1$ is the additional term added to the Hamiltonian at each iteration:
\begin{align}
    \bra{\psi}[U^\dagger (\thv^*) H(\paramh_j) U (\thv^*)- P_1 U^\dagger(\thv^*) H(\paramh_j) U (\thv^*) P_1 ]\ket{\psi} =& \bra{\psi}[U^\dagger (\thv^*) H(\paramh_{j-1}) U (\thv^*)- P_1 U^\dagger(\thv^*) H(\paramh_{j-1}) U (\thv^*) P_1 ]\ket{\psi} \\
    &+ \Delta\paramh\bra{\psi}[U^\dagger (\thv^*) H_1 U (\thv^*)- P_1 U^\dagger(\thv^*) H_1 U (\thv^*) P_1 ]\ket{\psi}  \\
    \leq & (1-\epsilon^2)\lambda_{\min} + \epsilon^2(\lambda_{\min}- \Delta_{\rm gap}) - \epsilon^2\lambda_{\min}\\
    & -(1-\epsilon^2)(\lambda_{\min}-\Delta_{\rm gap}) + |\Delta\paramh|\|H_1\|_s\\
    = & (1-2\epsilon^2)\Delta_{\rm gap} + |\Delta\paramh|\|H_1\|_s,
\label{eq:upper_bound_on_expval_simpler}
\end{align}
where in the inequality we used the explicit form of $\ket{\psi}$ described above, with the gap defined as $\Delta_{\rm gap} = \lambda_{\min} - \lambda_{1}$, with $\lambda_{\min}$ being the ground state energy and $\lambda_1$ being the first excited state energy. In the final equality we see that for sufficiently small $|\Delta\paramh|$ this quantity is negative, given that $\Delta_{\rm gap}<0$, assuming a good enough training in $\epsilon$ for the previous iteration. In particular, we consider $\epsilon^2\leq 1/2$. Indeed, if we choose 
\begin{align}
    |\Delta\paramh|\leq \widetilde{\gamma}\frac{(1-2\epsilon^2)|\Delta_{\rm gap}|}{\|H_1\|_s},
\end{align}
where $\widetilde{\gamma}\in(0,1)$ is a constant. Then, we can rewrite the upper bound in Eq.~\eqref{eq:upper_bound_on_expval_simpler} as 
\begin{align}
     \bra{\psi}[U^\dagger (\thv^*) H(\paramh_j) U (\thv^*)- P_1 U^\dagger(\thv^*) H(\paramh_j) U (\thv^*) P_1 ]\ket{\psi} \leq (1-2\epsilon^2)(1-\widetilde{\gamma})\Delta_{\rm gap} 
\end{align}
Finally, we can input this last equation in Eq.~\eqref{eq:h0_h1_difference_simpler} to obtain
\begin{align}
     \expval{ \left (H_{\rm eff}(\paramh_j) - P_1 H_{\rm eff}(\paramh_j) P_1 \right)}{\psi} \leq& K_+(1-2\epsilon^2)(1-\widetilde{\gamma})\Delta_{\rm gap}  + (1-K_+)\|H(\paramh_j)\|_s.
\end{align}
Imposing the negativity of this expression, we can derive the following bound on $K_+$
\begin{align}\label{eq:lowerbound_on_K_simpler}
    K_{+}^{(j)} > \frac{\|H(\paramh_j)\|_s}{\|H(\paramh_j)\|_s - (1-2\epsilon^2)(1-\widetilde{\gamma})\Delta_{\rm gap}}.
\end{align}
Finding a condition for $r$ such that the previous condition on $K_+$ is fulfilled is not trivial. However, we can leverage the Taylor Reminder Theorem~\ref{thm:taylor} to derive it. We start with the definition of $ K_{+}(r)$, where we make the dependence on $r$ explicit:
\begin{align}
    K_{+}(r) =&  \prod_{l=2}^{M}k_{+}(r)\\
    \geq & \prod_{l=2}^{M}\left( 1-\frac{r^2}{3}\right)\\
    \geq &\min_{l\in\{2,...,M\}}  \left( 1-\frac{r^2}{3}\right)^{M-1}\\
    \geq & 1-\frac{r^2(M-1)}{3} .\label{eq:lowerbound_on_K_wrt_r_simpler}
\end{align}
In the first inequality we used the Taylor Reminder Theorem in the function $k_{+}(r)$ defined in Eq~\eqref{eq:deffinition_kp_simpler}. In the second equality we used that the product of $M-1$ terms is minimized by taking the $M-1$ power of the smallest one. Finally, we use the Bernoulli's identity; $(1-x)^k\geq 1-xk$ if $ x\geq0$.

We can now input this condition in Eq.~\eqref{eq:lowerbound_on_K_simpler} to find a sufficient condition for it being negative:
\begin{equation}\label{eq:lowerbound_on_K_and_r_via_H_simpler}
    K_{+} \geq  1-\frac{r^2(M-1)}{3}  > \frac{\|H(\paramh_j)\|_s}{\|H(\paramh_j)\|_s - (1-2\epsilon^2)(1-\widetilde{\gamma})\Delta_{\rm gap}},
\end{equation}
where the first inequality is due to the previous derivation in Eq.~\eqref{eq:lowerbound_on_K_wrt_r_simpler}, and we impose the second inequality such that the condition in Eq.~\eqref{eq:lowerbound_on_K_simpler} is fulfilled. Using this, we obtain the following condition in $r$
\begin{align}
    1-\frac{r^2(M-1)}{3}  > \frac{\|H(\paramh_j)\|_s}{\|H(\paramh_j)\|_s - (1-2\epsilon^2)(1-\widetilde{\gamma})\Delta_{\rm gap}}\Rightarrow r^2<\gamma\frac{3}{(M-1)}\frac{(1-2\epsilon^2)(1-\widetilde{\gamma})|\Delta_{\rm gap}|}{\|H(\paramh_j)\|_s + (1-2\epsilon^2)(1-\widetilde{\gamma})|\Delta_{\rm gap}|}\label{eq:r_condition_simpler},
\end{align}
where $\gamma\in(0,1)$ is a constant. Now we can recover Eq.~\eqref{eq:lower_bound_variance_general_simpler}, and assuming that $r$ fulfills the condition in Eq.~\eqref{eq:r_condition_simpler}, we can further lower bound the variance as
\begin{align}
    \Var_{\alv\sim\DC(\vec{0},r)}[\LC_j(\alv)]\geq &  h(r) (\expval{  \left (H_{\rm eff}(\paramh_j) - P_1 H_{\rm eff}(\paramh_j) P_1 \right) }{\psi})^2 \\
    \geq & h(r) \left( (1-\gamma)(1-\widetilde{\gamma})(1-2\epsilon^2)|\Delta_{\rm gap}| \right)^2.\label{eq:final_var_tobound_simpler}
\end{align}

To finally give an intuitive expression, we lower bound the function $h(r)$ defined in Eq~\eqref{eq:h_function_simpler}. We do so by leveraging the Taylor Reminder Theorem~\ref{thm:taylor} once again. Particularly, we expand the function to the 6th order around $r=0$, and then show that this is always a lower bound. We use that the first to third derivatives, as well as the fifth derivative, are equal to zero, around $r = 0$. Thus, we can express the function as
\begin{align}
     h(r) =& \frac{1}{4!}\partial_{r'}^4 h(r') \bigg|_{r' =0}r^4 + \frac{1}{6!}\partial_{r'}^6h(r')\bigg|_{r'=0}r^6 +\order{r^8}\\
     =& \frac{4 r^4}{45}-\frac{16r^6}{315} + \order{r^8}. \label{eq:taylor_expand_h_simpler}
\end{align}

To show that the terms we wrote explicitly in Eq.~\eqref{eq:taylor_expand_h_simpler} are indeed a lower bound of the function $h(r)$ we recall the form of $h(r)$ in Eq.~\eqref{eq:h_function_simpler}. Then, we can find the Taylor Expansions of the trigonometric functions and use this to find the Taylor expansion of $h(r)$. The expansions of the aforementioned functions are
\begin{align}
    \sinc(4r) = &\sum_{n=0}(-1)^n\frac{(4r)^{2n}}{(2n+1)!}\\
    \cos(4r) = & \sum_{n=0}(-1)^n\frac{(4r)^{2n}}{(2n)!}.
\end{align}

The upper limit in the sum is $\infty$ for all the subsequent sums. Inserting these expansions into $h(r)$, we find
\begin{align}
    h(r) =&  \frac{1}{2}\sum_{n=1}(-1)^n\frac{(4r)^{2n-1}}{(2n+1)!} + \frac{1}{8} + \frac{1}{8}\sum_{n=0}(-1)^n\frac{(4r)^{2n}}{(2n)!}\\
    =& \frac{1}{8}\sum_{n=2}(-1)^n \frac{n-1}{n+1}\frac{(4r)^{2n}}{(2n+1)!}\label{eq:sum_hr_simpler}
\end{align}
where, as expected, in the second inequality we use that the powers of $r$ up to the third ($n<2$) cancel out. Furthermore, we compacted the sums into one. 

Note that this sum is absolutely convergent. Indeed, we can take the absolute value of all the elements:
\begin{align}
    \sum_{n=2}\left|\frac{n-1}{n+1}\right|\frac{(4r)^{2n}}{(2n+1)!}\leq \frac{1}{4r}\sum_{n=0}\frac{(4r)^{2n+1}}{(2n+1)!} = \frac{\sinh(4r)}{4r}
\end{align}
where in the first inequality we used that $|(n-1)/(n+1)|<1$ for $n > 0$ and in the last equality we used that the hyperbolic sine $\sinh (x)$ is has a Taylor expansion of that form.

Because the sum in Eq.~\eqref{eq:sum_hr_simpler} is absolutely convergent, we can rearrange the terms of the sum. Indeed, absolute convergent series can be rearranged and not only will they still converge, but they will converge to the same limit (see Remark 10.13 in Ref.~\cite{banaszczyk2006additive}). Thus, we rearrange Eq.~\eqref{eq:sum_hr_simpler} by pairing the elements $(n, n+1), (n+2,n+3)...$
\begin{align}
    h(r) =&\frac{1}{8}\sum_{n=2}(-1)^n \frac{n-1}{n+1}\frac{(4r)^{2n}}{(2n+1)!}\\
    =&\frac{1}{8}\sum_{n=1} \frac{(4r)^{4n}}{(4n+1)!}\left( \frac{2n-1}{2n+1} - \frac{2n}{2n+2}\frac{(4r)^{2}}{(4n+2)(4n+3)} \right).\label{eq:hr_rearanged}
\end{align}
In the second inequality we paired one positive term $n$ with the next $n+1$ and redefined $n\to 2n$ in order to avoid double counting due to this rearrangement. Now it is easy to show that for $r<1$, $n\geq1$ all this terms are always positive. Indeed we can rewrite the term as
\begin{align}
    \frac{2n-1}{2n+1}-  \frac{n}{n+1}\frac{16}{(4n+2)(4n+3)}
\end{align}
then it suffices to see that 
\begin{align}
   \begin{cases}
       \frac{n}{n+1}\frac{16}{(4n+2)(4n+3)}&\leq \frac{n}{n+1}\frac{8}{21}\\
       \frac{n}{n+1}\frac{8}{21} &\leq   \frac{2n-1}{2n+1}
   \end{cases} 
\end{align}
to be able to write the following inequality
\begin{align}
   \frac{2n-1}{2n+1}\geq \frac{n}{n+1}\frac{8}{21}\geq \frac{n}{n+1}\frac{16}{(4n+2)(4n+3)}.
\end{align}

Finally, we can lower bound the whole sum by one of the elements. We do this for $n=1$ and obtain 
\begin{align}
    h(r)\geq& \frac{1}{8}\frac{(4r)^{4}}{(4+1)!}\left( \frac{2-1}{2+1} - \frac{2}{2+2}\frac{(4r)^{2}}{(4+2)(4+3)}\right)\\
    =&\frac{1}{8} \frac{(4r)^{4}}{120}\left( \frac{2-1}{2+1} - \frac{2}{2+2}\frac{(4r)^{2}}{(4+2)(4+3)}\right)\\
    =&\left( 1 - \frac{4r^2}{7} \right)\frac{4r^4}{45}.
\end{align}

Therefore, we can recover Eq.~\eqref{eq:final_var_tobound_simpler} and input the result obtained in Eq.~\eqref{eq:taylor_expand_h_simpler} to find that
\begin{align}
    \Var_{\alv\sim\DC(\vec{0},r)}[\LC_j(\alv)]\geq & \left(  1-\frac{4r^2}{7} \right) \frac{4 r^4}{45} \left( (1-\gamma)(1-\widetilde{\gamma})(1-2\epsilon^2)|\Delta_{\rm gap}| \right)^2,
\end{align}
where $r$ should fulfill the constraint in Eq.~\eqref{eq:r_condition_simpler}. 
Because $r\in\Theta({\rm poly}(n))$, then $\Var_{\alv\sim\DC(\vec{0},r)}[\LC(\alv)]\in\Omega({\rm poly}(n))$.

\end{proof}

\section{Lower bound on the variance of the Meta-VQE loss function} \label{app:lowerbound_metavqe}
In this appendix, we derive a lower bound on the variance of a loss function of the form
\begin{equation}
    \LC(\thv, \boldsymbol{x}) = \frac 1 N \sum_{j = 1}^N\expval{U^\dagger (\thv, x_j) H(\paramh_j) U(\thv, x_j)}{\psi} = \frac 1 N \sum_{j = 1}^N \LC_j(\thv, x_j)\,. 
    \label{eq:loss_def_app}
\end{equation}
Unlike the previous case, the parameterized quantum circuit now includes an explicit encoding 
of the classical variable $x_j$ in its structure.

The iterative training procedure proceeds as follows. 
We begin with the single-term loss $\mathcal{L}(\boldsymbol{\theta}, x_1) 
= \mathcal{L}_1(\boldsymbol{\theta}, x_1)$ and train the circuit to obtain an optimal set 
of parameters $\boldsymbol{\theta}_1^*$. 
Next, we consider the two-term loss 
$\mathcal{L}(\boldsymbol{\theta}, (x_1, x_2)) 
= \tfrac{1}{2}\big[\mathcal{L}_1(\boldsymbol{\theta}, x_1)
+ \mathcal{L}_2(\boldsymbol{\theta}, x_2)\big]$ 
and re-optimize the circuit by initializing the parameters near $\boldsymbol{\theta}_1^*$. 
This process is repeated iteratively until step~$N$, each time adding one more term to the 
loss function and refining the parameters accordingly.

Importantly, the global loss $\mathcal{L}(\boldsymbol{\theta}, \boldsymbol{x})$ 
changes at every iteration because a new encoded Hamiltonian term 
is incorporated at each step, thus modifying both the optimization landscape 
and the effective variance of the loss.

\begin{theorem}
    Consider a loss function $\mathcal{L}(\thv, \vec x)$ of the form of 
    \begin{align}
    \label{eq:loss_encoding_th}
        \mathcal{L}(\thv, \vec{x}) = \frac 1N \sum_{j=1}^N \mathcal{L}^{(j)}(\vec{\theta}, \paramh_j)
    \end{align}
    where
    \begin{align}
        \LC^{(j)}(\thv,\paramh_j) = \expval{U^\dagger (\thv,\paramh_j) H(\paramh_j) U(\thv,\paramh_j)}{\psi}\,, 
    \end{align}
    and with the generators being non parametrized gates and Pauli rotations of the form of
    \begin{align}
        e^{-i\th_j g_j(\paramh)P_j}.
    \end{align}
    Assume that every encoding function $g_j$ is differentiable on the relevant interval has bounded derivatives. 

 Assume moreover that the first circuit generator acts non-trivially on the input state in the sense that
\begin{equation}
\langle\psi|P_1|\psi\rangle = 0.
\end{equation}

 Let $\boldsymbol{\theta}^*$ denote the parameters carried over from the previous iteration (trained on $x_1,\dots,x_{N-1}$).Assume that for each $j\in\{1,\dots,N-1\}$ the prepared state has overlap
\begin{equation}
\big|\langle {\rm gs}_j|U(\boldsymbol{\theta}^*,x_j)|\psi\rangle\big|^2\ge 1-\epsilon^2,
\qquad \epsilon\le \frac {1}{ \sqrt2},
\end{equation}
where $|{\rm gs}_j\rangle$ is a ground state of $H(x_j)$. Assume further that each $H(x_j)$ has a nonzero spectral gap
$\Delta_{\rm gap}^{(j)}$ between its ground energy and first excited energy, and define
\begin{equation}
|\Delta_{\rm gap\,min}|:=\min_{1\le j\le N}|\Delta_{\rm gap}^{(j)}|.
\end{equation}

 Then, if we add to the Hamiltonian a term of $\Delta \paramh W$ (that is $H(x_N)=H(x_{N-1})+\Delta x\,H_1$), where $H_1$ is a fixed hermitian matrix, and $\Delta\paramh = \paramh_N-\paramh_{N-1}$, such that it is bounded as 
    \begin{align}
        |\Delta\paramh|< \widetilde{\gamma} \frac{(1-2\epsilon^2)|\Delta_{\rm gap}^{(N-1)}|}{\|H_1\|_s + M \max_l \partial_\paramh g_l(\paramh) \|H(\paramh_N)\|_s},
    \end{align}
    where $\widetilde{\gamma}\in(0,1)$ is a constant, $M$ is the number of trainable parameters and $\|\cdot\|_s = \lambda_{\max}^{(\cdot)} -  \lambda_{\min}^{(\cdot)}$ is a semi-norm~\cite{boixo2007generalized}.

    Let $\mathcal{D}(\boldsymbol{\theta}^*,r)$ be the uniform distribution over the hypercube
$\{\boldsymbol{\theta}^*+\boldsymbol{\alpha}:\alpha_i\in[-r,r]\}$ (with i.i.d.\ coordinates), and let us consider some value of the radius $r_t$ satisfy
    \begin{align}
        r_t^2\leq \gamma \min\left\{\min_{j,k}\frac{4 g_1^2(\paramh_j)g_1^2(\paramh_k)}{45 h_6(\paramh_j,\paramh_k)}, \frac{3}{g_{\min}^2(M-1)}\frac{\widetilde{\gamma}(1-2\epsilon^2)|\Delta_{\rm gap\, min}|}{\|H\|_{s,{\rm min}} + (1-2\epsilon^2)|\Delta_{\rm gap\, min}|}\right\}
    \end{align}
    where 
    \begin{align}
       h_6(\paramh_j,\paramh_k)=& 2^6\Bigg( \frac{3  g_{1}(\paramh_j)^6}{7}+ \frac{ g_{1}(\paramh_j)^5  g_{1}(\paramh_k)}{2}+ \frac{37  g_{1}(\paramh_j)^4  g_{1}(\paramh_k)^2}{7}+ \frac{5 g_{1}(\paramh_j)^3 g_{1}(\paramh_k)^3}{4}\\
    &+ \frac{37  g_{1}(\paramh_j)^2  g_{1}(\paramh_k)^4}{7}+ \frac{ g_{1}(\paramh_j)  g_{1}(\paramh_k)^5}{7}+ \frac{ g_{1}(\paramh_j)  g_{1}(\paramh_k)^5}{2}+ \frac{3  g_{1}(\paramh_k)^6}{7}\Bigg) 
    \end{align}

    Then the variance on a hyper cube of side $2r$ with $0< r\leq r_t$ is lower-bounded by 
    \begin{align}
        \Var_{\thv\sim\DC(\thv^*,r)}[\LC(\alv)]\geq (1-\gamma)\frac{4}{45}g_1^2(\paramh_j)g_1^2(\paramh_k)(1-\gamma)(1-\widetilde{\gamma})(1-2\epsilon^2)|\Delta_{\rm gap \, min}| \in\Omega\left( 1/{\rm Poly}(n) \right),
    \end{align}
where $\thv^*$ is the previous iteration best solution and where $\gamma\in(0,1)$ is a constant.

\end{theorem}
\begin{proof}
Since we have already bounded each one of the terms $\LC_j(\thv, x_j)$ (but without the encoding $x_j$) in Appendix~\ref{app:noencoding}, we have much of the work done.

The parameterized quantum circuit in this case is given by 
    \begin{align}\label{eq:circuit_appendix}
        U(\thv,\paramh_j) = \prod_{i=1}^M  U_i(f_i(\th_i,\paramh_j))V_i\,,
    \end{align}
    where $\thv$ is a vector with the variational parameters $\thv = (\th_1, \th_2,...,\th_M)$ and $f_i(\th_i,\paramh_j)$ are the encoding functions. These functions can be different for every $\theta_i$ but share a common structure $f_i(\th_i,\paramh_j) = g_i(\paramh_j) \th_i$. Finally, each unitary gate takes the form of $U_k(\th_k,\paramh_j) = e^{-i f_k(\th_k,\paramh_j)P_k}$, where $P_k$ is a Pauli string. 

   As in Appendix~\ref{app:noencoding}, the parametrized gates can be re-expressed as perturbations $\alv$ around the previous optimal solution, i.e. $\th_i = \th_i^*+\al_i$ as shown below
    \begin{align}
        U(\thv,\paramh_j) =& \prod_{i=1}^M V_i U_i(f_i(\th_i,\paramh_j)) \\
        = & \prod_{i=1}^M V_i U_i(f_i(\th_i^*,\paramh_j))U_i(f_i(\al_i,\paramh_j)) \\
        =&  \prod_{i=1}^M \widetilde{V}_i(f_i(\th_i^*,\paramh_j))U_i(f_i(\al_i,\paramh_j))\, , 
    \end{align}
    where we used that $f_i(\th_i^* + \al_i, \paramh_j) = f_i(\th_i^*, \paramh_j) + f_i(\al_i, \paramh_j)$ and the fact that $e^{-i (a+b)P} =e^{-i aP}e^{-i bP} $ in the first equality to separate the parametrized unitary into two. In the final equality we absorbed the $\theta^*$ dependent $U_i$ into $V_i$, i.e. $\widetilde{V}_i(f_i(\th_i^*,\paramh_j)) = V_iU_i(f_i(\th_i^*,\paramh_j)) $. To ease the notation, we will denote $\widetilde{V}_i(f_i(\th_i^*,\paramh_j))$ as just $\widetilde{V_i}$.

    We consider a region of parameters around the previous solution of the form
    \begin{align}
        \VC(\thv^*,r) = \{\thv = \thv^* + \alv\,|\,\al_i\in[-r,r]\}
    \end{align}
    where $r$ is the characteristic length of the region. 

    We are interested in the variance of the loss function in Eq.~\eqref{eq:loss_encoding_th} over the region $\VC(\thv^*,r)$ such that the parameters are uniformly sampled 
    \begin{align}
        \DC(\thv^*,r) := {\rm Unif}[ \VC(\thv^*,r)]\,,
    \end{align}
    or in other words we are interested in
    \begin{align}\label{eq:variance_of_loss_app_exact}
     \Var_{\thv\sim\DC(\thv^*, r)} [\LC(\thv, \boldsymbol{x})] = &\Var_{\alv\sim\DC(\vec{0}, r)} [\LC(\thv=\thv^* + \alv, \boldsymbol{x})]\,.
     \end{align}

     Using this we now move to finding a general lower bound for the variance of the loss function.

 As we are interested in independently sampled parameters from the same distribution, we can use Proposition~\ref{prop:var_decomp} to immediately find a lower bound on Eq.~\eqref{eq:variance_of_loss_app_exact}
\begin{align}\label{eq:lowerbound_variance}
    \Var_{\alv\sim\DC(\vec{0},r)}[\LC(\alv, \boldsymbol{x})]\geq \Var_{\al_1\sim\DC(0,r)}[\Ebb_{\alv_{\overline{1}}\sim\DC(\vec{0}_{\overline{1}},r)}[\LC(\alv, \boldsymbol{x})]]\,
\end{align}
where we denote $\alv_{\overline{1}} = (\al_2,...\al_M)$, and similarly, $\vec{0} = {0,...,0}$ is the $0$-th vector of $M-1$ components.

We start by computing the expected values of Eq.~\eqref{eq:lowerbound_variance}. Because the variables are independent, we can compute the different expected values one-by-one, starting with $\al_2$ to $\al_{M}$, that is
\begin{align}
    \Ebb_{\alv_{\overline{1}}\sim\DC(\vec{0}_{\overline{1}},r)}[\LC(\alv, \boldsymbol{x})] = & \Ebb_{\al_2}[... \Ebb_{\al_M}[\LC(\alv, \boldsymbol{x})]...]\,.
\end{align}

As in Appendix~\ref{app:noencoding} we proceed to compute the expectation value with respect $\alpha_M$. We first compute this quantity for a single term of the loss function. The prodedure here it is very similar as in the previous section, but we include it for completeness.

To do this computation, we express each parametrized unitary $U_i(g_i(\paramh_j)\al_i) = e^{-i g_i(\paramh_j)\al_i P_i} = \cos(g_i(\paramh_j)\al_i)\1  -i\sin(g_i(\paramh_j)\al_i)P_i$. Then we can define
\begin{align}
    k_{+}^{(i,j)} :=& \frac{1}{2r}\int_{-r}^{r}\cos^2(g_i(\paramh_j)\al_i) = \frac{1}{2}(1+\sinc(2g_i(\paramh_j)r))\label{eq:deffinition_kp}\\
    k_{-}^{(i,j)} :=& \frac{1}{2r}\int_{-r}^{r}\sin^2(g_i(\paramh_j)\al_i) = \frac{1}{2}(1-\sinc(2g_i(\paramh_j)r))
\end{align}
where now we have index $i$ refering to layer and $j$ to encoding. We have explicitly used that $f_i(\th_i,\paramh_j) = g_i(\paramh_j) \th_i$. Equipped with these definitions, we can now proceed with the calculations. We start by evaluation the expected value of the unitary gate $M$ into one of the summand in Eq.~\eqref{eq:loss_encoding_th}, i.e. in the term in $\mathcal L^{(j)}(\thv , x_j)$. The expected value over the first parameter takes the following form
\begin{align}
    \Ebb_{\al_M\sim\DC(0,r)}[ U_M^\dagger(g_M(\paramh_j)\al_M)\widetilde{V}_M^\dagger H(\paramh_j)\widetilde{V}_M U_M(g_M(\paramh_j)\al_M)] =& \Ebb_{\al_M\sim\DC(0,r)}[\cos^2(g_M(x_j)\al_M)\widetilde{V}_M^\dagger H(\paramh_j)\widetilde{V}_M  \\
    &+\sin^2(g_M(x_j)\al_M)P_M\widetilde{V}_M^\dagger H(\paramh_j)\widetilde{V}_M P_M \\
    &- i\sin(g_M(x_j)\al_M)\cos(g_M(x_j)\al_M)[P_M, \widetilde{V}_M^\dagger H\widetilde{V}_M] ]\label{eq:expansion_in_sin_cos}\\
    =&  k_{+}^{(M,j)} \widetilde{V}_MH(\paramh_j)\widetilde{V}_M^\dagger +  k_{-}^{(M,j)} P_M  \widetilde{V}_M ^\dagger H(\paramh_j)\widetilde{V}_M P_M.
    \label{eq:expval_onegate}
\end{align}
Iin the first equality we expanded the exponential into its trigonometric form. Then, in the second equality, we used that the integral with an odd number of powers of $\sin$ over a symmetric space vanishes. Because all of the variables are independent and follow the same distribution, we reproduce the calculation in Eq.~\eqref{eq:expval_onegate} for all the subsequent parameters. To do so, we denote $H' :=  k_{+}^{(M,j)} V_M H(\paramh_j) V_M^\dagger +  k_{-}^{(M,j)} V_M P_M H(\paramh_j) P_M V_M^\dagger$, and using this notation we can see that
\begin{align}
    &\Ebb_{\al_{M-1}\sim\DC(0,r)}\left[ \left(\prod_{s=1}^{M-1} \widetilde{V}_s U_s(g_s(\paramh_j)\th_s) \right)^\dagger H'(\paramh_j)\left(\prod_{s=1}^{M-1} \widetilde{V}_s U_s(g_s(\paramh_j)\th_s) \right) \right]\\
    &= \left(\prod_{s=1}^{M-2} \widetilde{V}_s U_s(g_s(\paramh_j)\th_s) \right)^\dagger \big(k_{+}^{(M-1,j)} \widetilde{V}_{M-1}^\dagger H'(\paramh_j) \widetilde{V}_{M-1}  \\
    & \qquad +  k_{-}^{(M-1,j)} P_{M-1}  \widetilde{V}_{M-1}^\dagger H'(\paramh_j) \widetilde{V}_{M-1}P_{M-1}\bigg) \left( \prod_{s=1}^{M-1} \widetilde{V}_s U_s(g_s(\paramh_j)\th_s) \right) \\
    \label{eq:exp_value_M-1}
    &= \left(\prod_{s=1}^{M-2} \widetilde{V}_s U_s(g_s(\paramh_j)\th_s) \right)^\dagger  H''(\paramh_j)\left(\prod_{s=1}^{M-2} \widetilde{V}_s U_s(g_s(\paramh_j)\th_s) \right)
\end{align}
where we denoted 
\begin{align}
    H''(\paramh_j) =& k_+^{(M,j)}k_+^{(M-1,j)}\widetilde{V}_{M-1}^\dagger\widetilde{V}_{M}^\dagger H(\paramh_j)\widetilde{V}_{M}\widetilde{V}_{M-1} + k_+^{(M,j)}K_-^{(M-1,j)}P_{M-1}\widetilde{V}_{M-1}\widetilde{V}_{M}H(\paramh_j)\widetilde{V}_{M}^\dagger \widetilde{V}_{M-1}^\dagger P_{M-1}  \\
    &+ k_-^{(M,j)}k_+^{(M-1,j)}\widetilde{V}^\dagger_{M-1}P_{M}\widetilde{V}^\dagger_{M}H(\paramh_j)\widetilde{V}_{M}P_{M}\widetilde{V}_{M-1}  \\
    &+ k_+^{(M,j)}k_+^{(M-1,j)}P_{M-1}\widetilde{V}_{M-1}^\dagger P_{M}\widetilde{V}_{M}^\dagger H(\paramh_j)\widetilde{V}_{M}P_{M} \widetilde{V}_{M-1}P_{M-1}
\end{align}
We define as $H_{\rm eff}(x_j)$ the result of integrating all the remaining variables in Eq~\eqref{eq:exp_value_M-1}: 
\begin{equation}
   H_{\rm eff}(x_j):=  H\Ebb_{\al_{2}\sim\DC(0,r)}\left[\cdots  \Ebb_{\al_{M-2}\sim\DC(0,r)} \left[  \left(\prod_{s=1}^{M-2} \widetilde{V}_s U_s(g_s(\paramh_j)\th_s) \right)^\dagger  H''(\paramh_j)\left(\prod_{s=1}^{M-2} \widetilde{V}_s U_s(g_s(\paramh_j)\th_s) \right)\right] \cdots \right].
\end{equation}
Explicitly, $H_{\rm eff}$ takes the following form
\begin{align}\label{eq:expectedval_hj}
    H_{\rm eff}(\paramh_j) = \widetilde{V}_1^\dagger \sum_{\vec{z}}\bigcirc_{l=2}^M k_{z_l}^{(l,j)}\mathcal{P}_{z_l}^{(l)} (H(\paramh_j)) \widetilde{V}_1
\end{align}
where $z_l\in\{+,-\}$, $\bigcirc$ represents the composition operation ($f\circ g = f(g)$), and the super-operator $\mathcal{P}_{z}^{(l)}(\cdot) $ is defined as 
\begin{align}
    \mathcal{P}_{z}^{(l)}(\cdot) =
    \begin{cases}
        \widetilde{V}_l^\dagger \cdot \widetilde{V}_l {\, \rm if \, } z\to +\\
       P_l \widetilde{V}_l^\dagger \cdot  \widetilde{V}_l P_l {\, \rm if \, } z\to -
    \end{cases}
\end{align}
in analogy to Eq.~\eqref{eq:expectedval_hj_simpler}. 

Since the expected value is linear with the sum (i.e. $\Ebb[a+b] = \Ebb[a] + \Ebb[b]$), we can now use Eq.~\eqref{eq:expectedval_hj} to trivially compute the expected value of the loss function $\LC(\thv)$. Indeed, we can recover the explicit formula of the loss function and use it to express the expected value in Eq.~\eqref{eq:lowerbound_variance}:
\begin{align}
    \Ebb_{\alv_{\overline{1}}\sim\DC(\vec{0}_{\overline{1}},r)}[\LC(\thv)] =& \frac{1}{N}\sum_{j=1}^N \Ebb_{\alv_{\overline{1}}\sim\DC(\vec{0}_{\overline{1}},r)}[\LC^{(j)}(\thv,\paramh_j)]\\
    =& \frac{1}{N}\sum_{j=1}^N \expval{U_1^\dagger(g_1(\paramh_j)\al_1)\widetilde{V}_1^\dagger H_{\rm eff}(\paramh_j)\widetilde{V}_1U_1(g_1(\paramh_j)\al_1) }{\psi}
\end{align}
where in the last equality we used the equation derived in Eq.~\eqref{eq:expectedval_hj}. Furthermore, we stress that the dependency on still one variable is due to the fact that we still need to compute the variance over it.

That is what we proceed to do next. Using the previous equation, we are ready to compute the right hand side of Eq.~\eqref{eq:lowerbound_variance}, a lower bound on the variance over the constrained landscape. To do it, we start by expanding the gate $U_1$ into the $\sin,\cos$ form as we did in Eq.~\eqref{eq:expansion_in_sin_cos}
\begin{align}
     \Var_{\al_1\sim\DC(0,r)}[\Ebb_{\alv_{\overline{1}}\sim\DC(\vec{0}_{\overline{1}},r)}[\LC(\alv)] \geq &\frac{1}{N^2}\Var_{\al_1\sim\DC(0,r)}\Bigg[\sum_{j=1}^N \bra{\psi}\big(\cos^2(g_1(\paramh_j)\th_1)  H_{\rm eff} (\paramh_j)   \\      &+\sin^2(g_1(\paramh_j)\th_1) P_1  H_{\rm eff} (\paramh_j)  P_1 \big)\ket{\psi}\Bigg]  \\
     &+ \frac{1}{N^2}\Var_{\al_1\sim\DC(0,r)}\Bigg[\sum_{j=1}^N \sin(g_1(\paramh_j)\th_1)\cos(g_1(\paramh_j)\th_1) \expval{ i[P_M,  H(\paramh_j)] } {\psi} \Bigg]\label{eq:sincos_term_todrop}\\
     \geq &\frac{1}{N^2}\Var_{\al_1\sim\DC(0,r)}\Bigg[\sum_{j=1}^N \bra{\psi}\bigg(\cos^2(g_1(\paramh_j)\th_1) H_{\rm eff}(\paramh_j)  \\      &+\sin^2(g_1(\paramh_j)\th_1) P_1 H_{\rm eff}(\paramh_j) P_1\bigg) \ket{\psi}\Bigg] \label{eq:lowerbound_sincos_out}
\end{align}
where we expanded the gate $U_1$ using that $U_i(g_i(\paramh_j)\al_i) = e^{-i g_i(\paramh_j)\al_i P_i} = \cos(g_i(\paramh_j)\al_i)\1  -i\sin(g_i(\paramh_j)\al_i)P_i$ similarly to what we did in Eq.~\eqref{eq:expansion_in_sin_cos}. Furthermore, we used that the term that proportional to $\sin(g_1(\paramh_j)\th_1)\cos(g_1(\paramh_j)\th_1)$ is orthogonal to the terms $\sin^2(g_1(\paramh_j)\th_1),\,\cos^2(g_1(\paramh_j)\th_1)$. Therefore, the covariance between the two is zero. In the first inequality in Eq.~\eqref{eq:lowerbound_sincos_out}, we use that the variance of a function is always positive. Therefore we can drop the variance in Eq.~\eqref{eq:sincos_term_todrop}, as it is positive.

We now define the function $h(r, \paramh_j,\paramh_k)$ 
\begin{align}
    h(\paramh_j,\paramh_k) = \frac{1}{4}\left(\sinc[2\left(g_1(\paramh_j) - g_1(\paramh_k)\right)r]+ \sinc[2\left(g_1(\paramh_j) + g_1(\paramh_k)\right)r] -\frac{1}{2}\sinc[2g_1(\paramh_j)r]\sinc[2g_1(\paramh_k)r] \right)\label{eq:h_function}
\end{align}
which corresponds to the respective variances and covariances that we need to compute in Eq.~\eqref{eq:lowerbound_sincos_out}. When obvious by context, we are going to skip the dependence on $r$. In particular, $h(\paramh_j, \paramh_k)$ relates to these functions as 
\begin{align}
   h(\paramh_j,\paramh_k) = 
       &\Covar_{\al_1\sim\DC(0,r)}[\sin^2(g_1(\paramh_j)\th_1),\sin^2(g_1(\paramh_k)\th_1)]   \\=
   &\Covar_{\al_1\sim\DC(0,r)}[\cos^2(g_1(\paramh_j)\th_1),\cos^2(g_1(\paramh_k)\th_1)] \\ =& 
    - \Covar_{\al_1\sim\DC(0,r)}[\cos^2(g_1(\paramh_j)\th_1),\sin^2(g_1(\paramh_j)\th_1)]
  \label{eq:values_of_h_as_covars}
\end{align}
where we define the covariance in the standard way $\Covar_{\al\sim\DC}[\varphi_1(\al)\varphi_2(\al)] = \Ebb_{\al\sim\DC}[\varphi_1(\al)\varphi_2(\al)] - \Ebb_{\al\sim\DC}[\varphi_1(\al)]\Ebb_{\al\sim\DC}[\varphi_2(\al)]$.

Using these results, we can recover Eq.~\eqref{eq:lowerbound_sincos_out} and use the new introduced $ h(\paramh_j,\paramh_k)$ to compute the variance. We start by rewriting the variance as a sum of variances and covariances
\begin{align}
   &\frac{1}{N^2}\Var_{\al_1\sim\DC(0,r)}[\sum_{j=1}^N \bra{\psi}\cos^2(g_1(\paramh_j)\th_1)   H_{\rm eff}(\paramh_j)  +\sin^2(g_1(\paramh_j)\th_1) P_1   H_{\rm eff}(\paramh_j)P_1  \ket{\psi}] \\
    =&\frac{1}{N^2}\sum_{j,k}^N\Big(\Covar_{\al_1\sim\DC(0,r)}[\cos^2(g_1(\paramh_j)\th_1),\cos^2(g_1(\paramh_k)\th_1)]\expval{  H_{\rm eff}(\paramh_j) }{\psi}\expval{  H_{\rm eff}(\paramh_k) }{\psi}\\
    &+\Covar_{\al_1\sim\DC(0,r)}[\sin^2(g_1(\paramh_j)\th_1),\sin^2(g_1(\paramh_k)\th_1)]\expval{  P_1 H_{\rm eff}(\paramh_j)P_1  }{\psi}\expval{  P_1 H_{\rm eff}(\paramh_k)P_1  }{\psi}\\
    &+ \Covar_{\al_1\sim\DC(0,r)}[\sin^2(g_1(\paramh_j)\th_1),\cos^2(g_1(\paramh_k)\th_1)]\expval{  P_1 H_{\rm eff}(\paramh_j)P_1  }{\psi}\expval{  H_{\rm eff}(\paramh_k) }{\psi}\\
    &+ \Covar_{\al_1\sim\DC(0,r)}[\cos^2(g_1(\paramh_j)\th_1),\sin^2(g_1(\paramh_j)\th_1)]\expval{  H_{\rm eff}(\paramh_k) }{\psi}\expval{  P_1 H_{\rm eff}(\paramh_k)P_1  }{\psi}\Big)\\
    =&\frac{1}{N^2}\sum_{j,k}^N h(\paramh_j,\paramh_k)\left( \expval{  H_{\rm eff}(\paramh_j) }{\psi} - \expval{  P_1 H_{\rm eff}(\paramh_j) P_1  }{\psi} \right)\left( \expval{  H_{\rm eff}(\paramh_k) }{\psi} - \expval{  P_1 H_{\rm eff}(\paramh_k) P_1  }{\psi} \right)
    \label{eq:final_calcul_for_lowerbound}
\end{align}
where in the first equality we decomposed the variance of a sum into its covariances, i.e. $\Var[\varphi_1 + \varphi_2] = \Covar[\varphi_1,\varphi_1] + \Covar[\varphi_2,\varphi_2] + 2 \Covar[\varphi_1,\varphi_2]$, using the notation that $\Covar[\varphi_1,\varphi_1] = \Var[\varphi_1]$. In the second equality we used Eq.~\eqref{eq:values_of_h_as_covars}, to simplify the result. Finally we use $\cdot$ to emphasize the product of the two terms in the parenthesis due to the line break.

By recovering Eq.~\eqref{eq:lowerbound_variance}, and using the lower bound derived in Eq.~\eqref{eq:lowerbound_sincos_out} and the result of this term in Eq.~\eqref{eq:final_calcul_for_lowerbound}, we can finally derive a strict lower bound to the variance of the loss function that reads as follows
\begin{align}
    \Var_{\alv\sim\DC(\vec{0},r)}[\LC(\alv)]\geq & \frac{1}{N^2}\sum_{j,k}^N h(\paramh_j,\paramh_k) \expval{  \left (H_{\rm eff}(\paramh_j) - P_1 H_{\rm eff}(\paramh_j) P_1 \right) }{\psi} \label{eq:lower_bound_variance_general}\\
    &\cdot \expval{  \left(H_{\rm eff}(\paramh_k) - 
 P_1 H_{\rm eff}(\paramh_k) P_1 \right)  }{\psi} 
\end{align}

Finally, we want to make sure that this lower-bound is positive, and thus meaningful. 

\subsection{Sufficient conditions for the bound to be non-trivial}\label{sec:sufficient_condition_bound_meaningful}
Here we derive sufficient conditions for the bound to be meaningful. We start by analyzing the term $\expval{ \left (H_{\rm eff}(\paramh_j) - P_1 H_{\rm eff}(\paramh_j) P_1 \right)}{\psi}$. In particular, we want to leverage the known structure of $H_{\rm eff}(\paramh_j)$ to ensure that this term is always negative for all $j$. 

We first focus on all the terms with $j < N$. For these terms, the parameters have already been trained around the corresponding solution. As a consequence, we can invoke our assumption that the variational state has a sufficiently large overlap with the ground state of $H_{\rm eff}(\paramh_j)$, which in turn guarantees the negativity of the above expression, as we discuss now.

In contrast, when we add a new term to the loss function, i.e., the contribution with $j = N$, we do not yet have a guarantee on this overlap. For this last term we must instead rely on additional properties, such as the continuity of the encoding function and the smooth dependence of the optimal parameters on $\paramh_j$, to control the sign of the contribution. Note that the case $j = N$ always refers to the newest term introduced into the cost function at a given stage of the iterative procedure.

\subsubsection{Negativity of the expected values for \texorpdfstring{$j<N$}{j<N}}

We pont out that $H_{\rm eff}(\paramh_j)$ in Eq.~\eqref{eq:expectedval_hj} is a convex combination of rotated Hamiltonians, since the parameters $k_{+}^{(l,j)} + k_{-}^{(l,j)} = 1,\, \forall\, (l,j)$ and the differences in the Hamiltonians depend on unitary matrices exclusively. Furthermore, there is clearly one dominant term out, which is $K_+^{(j)} = \prod_{l=1}^{M-1} k_{+}^{(l,j)}$. This term is the largest for small values of $r$. Indeed, recalling Eq.~\eqref{eq:deffinition_kp} and by with a Taylor expansion, one can see that $k_{+}^{(l,j)}\sim 1 - \frac{g_l^2(\paramh_j)r^2}{12} + \order{r^4}$, $k_{-}^{(l,j)}\sim \frac{g_l^2(\paramh_j)r^2}{12}$. This means that around zero, we expect the term $K_+^{(j)}$ to be dominant. In the next lines we formalize this argument and find strict bounds on the value of $r$, such that the term $\expval{\widetilde{V}_1 \left (H_{\rm eff}(\paramh_j) - P_1 H_{\rm eff}(\paramh_j) P_1 \right)\widetilde{V}_1^\dagger}{\psi}$ is always negative. We start by rewriting this term as a sum of the terms depending on $K_+^{(j)}$ and those that do not
\begin{align}
    \expval{  \left (H_{\rm eff}(\paramh_j) - P_1 H_{\rm eff}(\paramh_j) P_1 \right) }{\psi} =& K_+^{(j)} \expval{U(\thv^*,\paramh_j) H(\paramh_j) U^\dagger (\thv^*,\paramh_j) }{\psi} \\
    &-  K_+^{(j)} \expval{  P_1   U(\thv^*,\paramh_j) H(\paramh_j) U^\dagger (\thv^*,\paramh_j)   P_1   }{\psi}\\
    &+\expval{   \widetilde{V}_1^\dagger \sum_{\vec{z}\neq \{+\}^{\otimes M}}\bigcirc_{l=2}^M k_{z_l}^{(l,j)}\mathcal{P}_{z_l}^{(l)} (H(\paramh_j)) \widetilde{V}_1 }{\psi}\\
    &- \expval{   P_1 \widetilde{V}_1^\dagger \sum_{\vec{z}\neq \{+\}^{\otimes M}}\bigcirc_{l=2}^M k_{z_l}^{(l,j)}\mathcal{P}_{z_l}^{(l)} (H(\paramh_j)) \widetilde{V}_1 P_1 }{\psi}\\
    =&  K_+^{(j)}\bra{\psi}[U(\thv^*,\paramh_j) H(\paramh_j) U^\dagger (\thv^*,\paramh_j) \\
    &-  P_1   U(\thv^*,\paramh_j) H(\paramh_j) U^\dagger (\thv^*,\paramh_j)   P_1  . ]\ket{\psi}  +\xi\label{eq:to_upperbound_xi}
\end{align}
In the first equality we used the definition of $H_{\rm eff}$ introduced in Eq.~\eqref{eq:expectedval_hj}. In the second equality we introduced $\xi$ a quantity that contains all the previous terms. As discussed, $H_{\rm eff}(\paramh_j)$ can be understood as a convex combination of rotated versions of the Hamiltonian $H(\paramh_j)$. Therefore, $\xi$ is upper bounded by the norm of $H(\paramh_j)$. This can be seen by applying H\"older's inequality as follows
\begin{align}
    |\xi|=&\left|\sum_{\vec{z}\neq \{+\}^{\otimes M}} \expval{ \widetilde{V}_1 ^\dagger \bigcirc_{l=2}^M k_{z_l}^{(l,j)}\mathcal{P}_{z_l}^{(l)} (H(\paramh_j))\widetilde{V}_1}{\psi} - \expval{  P_1 \widetilde{V}_1^\dagger\bigcirc_{l=2}^M k_{z_l}^{(l,j)}\mathcal{P}_{z_l}^{(l)} (H(\paramh_j)) \widetilde{V}_1 P_1}{\psi}\right|\\
    \leq&\sum_{\vec{z}\neq \{+\}^{\otimes M}} \prod_{l=2}^M k_{z_l}^{(l,j)} \max_{\vec{z}'\in\{+,-\}^{\otimes M}}\| \widetilde{V}_1^\dagger\bigcirc_{k=2}^M \mathcal{P}_{z'_k}^{(k)} (H(\paramh_j))\widetilde{V}_1 -  P_1\widetilde{V}_1^\dagger \bigcirc_{k=2}^M \mathcal{P}_{z_k}^{(k)} (H(\paramh_j)) \widetilde{V}_1 P_1\|_{\infty}\\
    \leq & \|H(\paramh_j)\|_{s} \sum_{\vec{z}\neq \{+\}^{\otimes M}} \prod_{l=2}^M k_{z_l}^{(l,j)} = (1-K_+^{(j)}) \|H(\paramh_j)\|_{s} \label{eq:bound_with_Hs_firstime}
\end{align}
where we define the semi-norm $\|H\|_{s}:=\lambda_{\max} - \lambda_{\min}$ (see Ref.~\cite{boixo2007generalized}). Using this result, we can upper bound equation Eq.~\eqref{eq:to_upperbound_xi}
\begin{align}
     \expval{ \left (H_{\rm eff}(\paramh_j) - P_1 H_{\rm eff}(\paramh_j) P_1 \right)}{\psi} \leq& K_+^{(j)}\bra{\psi}[U^\dagger (\thv^*,\paramh_j) H(\paramh_j) U (\thv^*,\paramh_j) \\
    &- P_1 U^\dagger(\thv^*,\paramh_j) H(\paramh_j) U (\thv^*,\paramh_j) P_1 ]\ket{\psi}  + (1-K_+^{(j)})\|H(\paramh_j)\|_s\label{eq:h0_h1_difference}
\end{align}

Next, we focus on the term proportional to $K_+^{(j)}$. To analyze it, we recall that $\thv^*$ are the set of variational parameters obtained when optimizing the previous loss landscape with $N-1$ terms. Therefore, if we assume that we trained well enough for all $j\neq N$, $\thv^*$ is a \textit{good enough} solution. Using this, we can write $U(\thv^*,\paramh_j)\ket{\psi} = \sqrt{1-\epsilon^2}\ket{{\rm gs}_{j}} + \epsilon \ket{{\rm gs}_{j}^\perp}$, where $\ket{{\rm gs}_{j}}$ is the ground state of $H(\paramh_j)$ and $\ket{{\rm gs}_{j}^\perp}$ is orthogonal to the ground state of $H(\paramh_j)$. Trivially, we can impose that $|\expval{P_1}{\psi}|^2=0$, as it just means that the first generator of the circuit acts non-trivially on the initial input state. This also implies that $\ket{\psi},P_1\ket{\psi}$ are orthogonal, and therefore we can write $U (\thv^*,\paramh_j)P_1\ket{\psi} = \epsilon\ket{{\rm gs}_{j}} -
 \sqrt{1-\epsilon^2}\ket{{\rm gs_{j}^\perp}}$, then we can upper bound the following term
\begin{align}
    \bra{\psi}[U^\dagger(\thv^*,\paramh_j) H(\paramh_j) U (\thv^*,\paramh_j)- P_1  U^\dagger(\thv^*,\paramh_j) H(\paramh_j) U (\thv^*,\paramh_j)  P_1  ]\ket{\psi}=& (1-\epsilon^2)\lambda_{\min} + (\lambda_{\min}-\Delta_{\rm gap})\epsilon^2 \\
    &- \epsilon^2\lambda_{\min} - (1-\epsilon^2)(\lambda_{\min}-\Delta_{\rm gap})\label{eq:worst_case_perp_state}\\
    =& (1-2\epsilon^2)\Delta_{\rm gap}^{(j)}\label{eq:exval_with_epsilon_val}
\end{align}
where we assume that the observable $H(\paramh_j)$ has a gap between the ground state and the first excited state, that we denoted via $\Delta_{\rm gap} = \lambda_{\min} - \lambda_1\leq0$, where $\lambda_{1}$ is the first excited level. Note that, instead of obtaining that $P_1$ maps the state $\ket{\psi}$ to a state of the form $U (\thv^*,\paramh_j)P_1\ket{\psi} = \epsilon\ket{{\rm gs}_{j}} -
 \sqrt{1-\epsilon^2}\ket{{\rm gs}_{j}^\perp}$, we could also obtain a state that is outside of the subspace spanned by $\left\{ \ket{{\rm gs}_j},\ket{gs_{j}^\perp} \right\}$. In such case, we would obtain a smaller value in Eq.~\eqref{eq:worst_case_perp_state}. Therefore, we just assume the worst case scenario.

Finally, we can recover Eq.~\eqref{eq:h0_h1_difference} and find a condition on $K_{+}^{(j)}$ and $\epsilon$ such that the quantity is always negative
\begin{align}\label{eq:lower_bound_expval_jleqN}
     \expval{ \left (H_{\rm eff}(\paramh_j) - P_1 H_{\rm eff}(\paramh_j) P_1 \right)}{\psi} \leq K_+^{(j)}(1-2\epsilon^2)\Delta_{\rm gap}^{(j)} + (1-K_+^{(j)})\|H(\paramh_j)\|_s.
\end{align}
We then need $\epsilon^2\leq 1/2$, and a  value of $K_{+}^{(j)}$ that fulfills the following:
\begin{align}\label{eq:lowerbound_on_K}
    K_{+}^{(j)} > \frac{\|H(\paramh_j)\|_s}{\|H(\paramh_j)\|_s - (1-2\epsilon^2)\Delta_{\rm gap}^{(j)}}.
\end{align}
Finding a condition for $r$ such that the condition on $K_+^{(j)}$ in this equation is fulfilled is not trivial. However, we can leverage the Taylor Reminder ( see Theorem~\ref{thm:taylor}) to derive it. To see this, we start with the definition of $ K_{+}^{(j)}(r)$ where we explicitly write the $r$ dependence:
\begin{align}
    K_{+}^{(j)}(r) =&  \prod_{l=2}^{M}k_{+}^{(l,j)}(r)\\
    \geq & \prod_{l=2}^{M}\left( 1-\frac{g_l^2(\paramh_j)r^2}{3}\right)\\
    \geq &\min_{l\in\{2,...,M\}}  \left( 1-\frac{g_l^2(\paramh_j)r^2}{3}\right)^{M-1}\\
    \geq & 1-\frac{g_{\max}^2(\paramh_j)r^2(M-1)}{3} \label{eq:lowerbound_on_K_wrt_r}.
\end{align}
In the first inequality we used the Taylor Reminder Theorem in the function $k_{+}^{(l,j)}(r)$ defined in Eq~\eqref{eq:deffinition_kp}, in the second equality we used that the product of $M-1$ terms is minimized by taking the power $M-1$ of the smallest one, and finally in the last inequality we used Bernoulli's identity; $(1+x)^k\geq 1-xk$ if $ x\geq0$. Also, we introduced the notation $g_{\max}(\paramh_j) := \max_l g_{l}(\paramh_j)$.

We can now input this condition in Eq.~\eqref{eq:lowerbound_on_K} to find a sufficient condition for it to be fulfilled
\begin{equation}\label{eq:lowerbound_on_K_and_r_via_H}
    K_{+}^{(j)} \geq  1-\frac{g_{\max}^2(\paramh_j)r^2(M-1)}{3}  > \frac{\|H(\paramh_j)\|_s}{\|H(\paramh_j)\|_s - (1-2\epsilon^2)\Delta_{\rm gap}^{(j)}},
\end{equation}
where the first inequality is due to the previous derivation in Eq.~\eqref{eq:lowerbound_on_K_wrt_r}, and the second inequality is imposed such that the condition in Eq.~\eqref{eq:lowerbound_on_K} is fulfilled.

Using the second inequality in Eq.~\eqref{eq:lowerbound_on_K_and_r_via_H} , we can easily derive a condition on $r$ such that Eq.~\eqref{eq:lowerbound_on_K} is satisfied. To do so, we simple isolate $r$ to obtain the following condition
\begin{align}
    1-\frac{g_{\max}^2(\paramh_j)r^2(M-1)}{3}  > \frac{\|H(\paramh_j)\|_s}{\|H(\paramh_j)\|_s - (1-2\epsilon^2)\Delta_{\rm gap}^{(j)}}\Rightarrow r< \frac{3}{g_{\max}^2(x_j)(M-1)}\frac{(1-2\epsilon^2)|\Delta_{\rm gap}^{(j)}|}{\|H(\paramh_j)\|_s + (1-2\epsilon^2)|\Delta_{\rm gap}^{(j)}|}.
\end{align}

Finally, if we want a condition that  works for all $j\leq N-1$, we need to find the tightest condition, i.e. 
\begin{align}\label{eq:condition_on_r_jleqN}
    r< \frac{3}{g_{\max}^2(M-1)}\frac{(1-2\epsilon^2)|\Delta_{\rm gap\, min}|}{\|H\|_{s,{\rm max}} + (1-2\epsilon^2)|\Delta_{\rm gap\, min}|}.
\end{align} 
Here we introduced some notation to ease the reading: $\|H\|_{s,{\rm max}} := \max_j \|H(\paramh_j)\|_{s}$, and $|\Delta_{\rm gap\, min}| := \min_j |\Delta_{\rm gap}^{(j)}|$.

We have now found a sufficient condition to ensure that the terms $j\neq N$ in Eq.~\eqref{eq:lower_bound_variance_general} are negative, and thus the lower-bound from such equation is not trivial. However, we need to extend this condition to the only term we have not considered yet, $j=N$. 

\subsubsection{Extending the condition to \texorpdfstring{$j=N$}{j=N}}

Finally, we need to extend this condition to $j = N$. Indeed, in this setting we cannot ensure that we have a \textit{good enough} fidelity with the ground state. To do so, we start from Eq.~\eqref{eq:h0_h1_difference}, as the analysis up to that point still applies. In particular, we bound the following term
\begin{align}
    \bra{\psi}[U^\dagger(\thv^*,\paramh_N) H(\paramh_N) U (\thv^*,\paramh_N)- P_1  U^\dagger(\thv^*,\paramh_N) H(\paramh_N) U (\thv^*,\paramh_N)  P_1  ]\ket{\psi}.
\end{align}

To do so, we start by adding and subtracting the following term $\bra{\psi}[U^\dagger(\thv^*,\paramh_{N-1}) H(\paramh_N) U (\thv^*,\paramh_{N-1})- P_1  U^\dagger(\thv^*,\paramh_{N-1}) H(\paramh_N) U (\thv^*,\paramh_{N-1})  P_1  ]\ket{\psi}$ to obtain
\begin{align}
    &\bra{\psi}[U^\dagger(\thv^*,\paramh_N) H(\paramh_N) U (\thv^*,\paramh_N)- P_1  U^\dagger(\thv^*,\paramh_N) H(\paramh_N) U (\thv^*,\paramh_N)  P_1  ]\ket{\psi}\\ \label{eq:togrpoup_in_jeqN2}   & =\bra{\psi}[U^\dagger(\thv^*,\paramh_N) H(\paramh_N) U (\thv^*,\paramh_N)- P_1  U^\dagger(\thv^*,\paramh_N) H(\paramh_N) U (\thv^*,\paramh_N)  P_1  ]\ket{\psi} \\
    &+ \bra{\psi}[U^\dagger(\thv^*,\paramh_{N-1}) H(\paramh_N) U (\thv^*,\paramh_{N-1})- P_1  U^\dagger(\thv^*,\paramh_{N-1}) H(\paramh_N) U (\thv^*,\paramh_{N-1})  P_1  ]\ket{\psi}\\
   & -\bra{\psi}[U^\dagger(\thv^*,\paramh_{N-1}) H(\paramh_N) U (\thv^*,\paramh_{N-1})- P_1  U^\dagger(\thv^*,\paramh_{N-1}) H(\paramh_N) U (\thv^*,\paramh_{N-1})  P_1  ]\ket{\psi} \label{eq:togrpoup_in_jeqN1}\\
   \leq & \bra{\psi}[U^\dagger(\thv^*,\paramh_{N-1}) H(\paramh_N) U (\thv^*,\paramh_{N-1})- P_1  U^\dagger(\thv^*,\paramh_{N-1}) H(\paramh_N) U (\thv^*,\paramh_{N-1})  P_1  ]\ket{\psi} \\
   &+ |\paramh_N-\paramh_{N-1}|\partial_{\paramh}(\bra{\psi}[U^\dagger(\thv^*,\paramh) H(\paramh_N) U (\thv^*,\paramh)- P_1  U^\dagger(\thv^*,\paramh) H(\paramh_N) U (\thv^*,\paramh)  P_1  ]\ket{\psi} ) \\
   \leq & \bra{\psi}[U^\dagger(\thv^*,\paramh_{N-1}) H(\paramh_N) U (\thv^*,\paramh_{N-1})- P_1  U^\dagger(\thv^*,\paramh_{N-1}) H(\paramh_N) U (\thv^*,\paramh_{N-1})  P_1  ]\ket{\psi}  \\
   &+ |\paramh_N-\paramh_{N-1}|\sum_i\partial_{g_i}(\bra{\psi}[U^\dagger(\thv^*,\paramh_N) H(\paramh_N) U (\thv^*,\paramh_N)- P_1  U^\dagger(\thv^*,\paramh_N) H(\paramh_N) U (\thv^*,\paramh_N)  P_1  ]\ket{\psi} )\partial_\paramh g_i(\paramh) .\label{eq:to_recover_jeqN_prior_to_derivative}
\end{align}
In the first inequality we grouped Eqs~(\ref{eq:togrpoup_in_jeqN2},\ref{eq:togrpoup_in_jeqN1}) to leverage that an expectation value parametrized with a circuit of the form of Eq.~\eqref{eq:circuit_appendix} is Lipschitz continuous (as long as the function $g$ is). Indeed, a function of this type is always differentiable and with a bounded derivative. In the final inequality, we apply the chain rule by first differentiating the expectation value with respect to each function $g_i$, and then differentiating each $g_i$ with respect to $\paramh$. Next, we compute the derivative of the expected value with respect to $g$, to further lower bound this quantity. We assume that at most $M$ gates in the ansatz are parametrized by $\paramh$, with parameters $\{g_i(\paramh)\}_{i=1}^M$. Then, we can write
\begin{align}
   & \sum_i^{M}\partial_{g_i}(\bra{\psi}[U^\dagger(\thv^*,\paramh_N) H(\paramh_N) U (\thv^*,\paramh_N)- P_1  U^\dagger(\thv^*,\paramh_N) H(\paramh_N) U (\thv^*,\paramh_N)  P_1  ]\ket{\psi} )\partial_\paramh g_i(\paramh)\\
    \leq &M \max_i| \partial_{g_i}\bra{\psi}[U^\dagger(\thv^*,\paramh_N) H(\paramh_N) U (\thv^*,\paramh_N)- P_1  U^\dagger(\thv^*,\paramh_N) H(\paramh_N) U (\thv^*,\paramh_N)  P_1  ]\ket{\psi} \partial_\paramh  g_i(\paramh)|\\
    \leq & M \max_i |\partial_{g_i}\bra{\psi}[U^\dagger(\thv^*,\paramh_N) H(\paramh_N) U (\thv^*,\paramh_N)- P_1  U^\dagger(\thv^*,\paramh_N) H(\paramh_N) U (\thv^*,\paramh_N)  P_1  ]\ket{\psi}| \max_l|\partial_\paramh  g_l(\paramh)| \\
    \leq & M \|H(\paramh_N)\|_s\max_l|\partial_\paramh  g_l(\paramh)|.
\end{align}
In the first inequality we bound the sum by $M$ times the maximum term. In the second we use $\max_i |a_i b_i| \leq (\max_i |a_i|)(\max_i |b_i|)$. Finally, in the last inequality we use that, for circuits of the form in Eq.~\eqref{eq:circuit_appendix}, the derivatives with respect to the gate parameters $g_i$ can be computed via the parameter-shift rule, and are bounded in magnitude by the seminorm $\|H(\paramh_N)\|_s$, as shown in Eq.~\eqref{eq:bound_with_Hs_firstime}.

We can now return to Eq.~\eqref{eq:to_recover_jeqN_prior_to_derivative}, and substitute the derivative by the above bound. Thus, we obtain
\begin{align}
    &\bra{\psi}[U^\dagger(\thv^*,\paramh_N) H(\paramh_N) U (\thv^*,\paramh_N)- P_1  U^\dagger(\thv^*,\paramh_N) H(\paramh_N) U (\thv^*,\paramh_N)  P_1  ]\ket{\psi}\\
    \leq &\bra{\psi}[U^\dagger(\thv^*,\paramh_{N-1}) H(\paramh_N) U (\thv^*,\paramh_{N-1})- P_1  U^\dagger(\thv^*,\paramh_{N-1}) H(\paramh_N) U (\thv^*,\paramh_{N-1})  P_1  ]\ket{\psi}\\
    &+|\paramh_N-\paramh_{N-1}|M \max_l|\partial_\paramh g_l(\paramh)| \|H(\paramh_N)\|_s \label{eq:bound_the_derivative_of_expval}\\
    = & \bra{\psi}[U^\dagger(\thv^*,\paramh_{N-1}) H(\paramh_{N-1}) U (\thv^*,\paramh_{N-1})- P_1  U^\dagger(\thv^*,\paramh_{N-1}) H(\paramh_{N-1}) U (\thv^*,\paramh_{N-1})  P_1  ]\ket{\psi}\label{eq:sepparate_h_in_old_and_sum1} \\
    &+ \bra{\psi}[U^\dagger(\thv^*,\paramh_{N-1}) \{H(\paramh_N)-H(\paramh_{N-1})\} U (\thv^*,\paramh_{N-1})\\
    &- P_1  U^\dagger(\thv^*,\paramh_{N-1}) \{H(\paramh_N)-H(\paramh_{N-1})\} U (\thv^*,\paramh_{N-1})  P_1  ]\ket{\psi} \label{eq:sepparate_h_in_old_and_sum2}\\
    &+|\paramh_N-\paramh_{N-1}|M \max_l|\partial_\paramh g_l(\paramh)| \|H(\paramh_N)\|_s,
\end{align}
where we expanded $H(\paramh_N) = H(\paramh_{N-1})+H(\paramh_N)- H(\paramh_{N-1})$. Then, using that $H(\paramh_N)-H(\paramh_{N-1}) = \Delta\paramh H_1$ we can upper bound the remaining terms as follows, starting by Eq.~\eqref{eq:sepparate_h_in_old_and_sum2}
\begin{align}
     &\bra{\psi}[U^\dagger(\thv^*,\paramh_{N-1}) \{H(\paramh_N)-H(\paramh_{N-1})\} U (\thv^*,\paramh_{N-1})\\
     &- P_1  U^\dagger(\thv^*,\paramh_{N-1}) \{H(\paramh_N)-H(\paramh_{N-1})\} U (\thv^*,\paramh_{N-1})  P_1  ]\ket{\psi}\\
     =&  \bra{\psi}[U^\dagger(\thv^*,\paramh_{N-1})  \Delta\paramh H_1 U (\thv^*,\paramh_{N-1})- P_1  U^\dagger(\thv^*,\paramh_{N-1})  \Delta\paramh H_1 U (\thv^*,\paramh_{N-1})  P_1  ]\ket{\psi}\\
     \leq & |\Delta\paramh|\|H_1\|_s\label{eq:sepparate_h_in_old_and_sum2_bound}
\end{align}
where we bounded the difference in operators in the same way that we did previously as in Eq.~\eqref{eq:bound_with_Hs_firstime}.

Next we focus on Eq.~\eqref{eq:sepparate_h_in_old_and_sum1}
\begin{align}
    \bra{\psi}[U^\dagger(\thv^*,\paramh_{N-1}) H(\paramh_{N-1}) U (\thv^*,\paramh_{N-1})- P_1  U^\dagger(\thv^*,\paramh_{N-1}) H(\paramh_{N-1}) U (\thv^*,\paramh_{N-1})  P_1  ]\ket{\psi}\leq (1-2\epsilon^2)\Delta_{\rm gap}^{(N-1)}\label{eq:sepparate_h_in_old_and_sum1_bound}
\end{align}
where we used the previously derived Eq.~\eqref{eq:exval_with_epsilon_val}.

Finally we can recover Eq.~\eqref{eq:sepparate_h_in_old_and_sum1} and apply the two bounds derived now in Eqs.~(\ref{eq:sepparate_h_in_old_and_sum2_bound},\ref{eq:sepparate_h_in_old_and_sum1_bound}) to upper bound it by 
\begin{align}
 &\bra{\psi}[U^\dagger(\thv^*,\paramh_N) H(\paramh_N) U (\thv^*,\paramh_N)- P_1  U^\dagger(\thv^*,\paramh_N) H(\paramh_N) U (\thv^*,\paramh_N)  P_1  ]\ket{\psi}\\
    \leq& (1-2\epsilon^2)\Delta_{\rm gap}^{(N-1)} + |\Delta\paramh|M \max_l|\partial_\paramh g_l(\paramh)| \|H(\paramh_N)\|_s + |\Delta\paramh|\|H_1\|_s|\\
    \leq & (1-2\epsilon^2)\Delta_{\rm gap}^{(N-1)} + |\Delta \paramh|M \max_l|\partial_\paramh g_l(\paramh)| (\|H(\paramh_{N-1})\|_s + \|H_1\|_s|\Delta\paramh|) + |\Delta\paramh|\|H_1\|_s.
\label{eq:bound_expval_jeqN}
\end{align}
In the last inequality we expressed $H(\paramh_N) = H(\paramh_{N-1})+\Delta\paramh H_1$ and applied triangular inequality to split the term into two: $\|H(\paramh_N)\| \leq \|H(\paramh_N+1)\|_s+|\Delta\paramh| \|H_1\|_s$. Note that this is not straightforward as $\|\cdot\|_s$ is not a norm. However, we can show that $\lambda_{\max}[H(\paramh_N)]< \lambda_{\max}[ H(\paramh_{N-1})] + \lambda_{\max}[\Delta\paramh H_1]$ and conversely  $\lambda_{\min}[H(\paramh_N)]>\lambda_{\min}[ H(\paramh_{N-1})] + \lambda_{\min}[\Delta\paramh H_1]$. With this, we can write 
\begin{align}
    |\lambda_{\max}[H(\paramh_N)]-\lambda_{\min}[H(\paramh_N)]|\leq& | \lambda_{\max}[ H(\paramh_{N-1})] + \lambda_{\max}[\Delta\paramh H_1] - \lambda_{\min}[ H(\paramh_{N-1})] - \lambda_{\min}[\Delta\paramh H_1] |\\
    = & | (\lambda_{\max}[ H(\paramh_{N-1})] - \lambda_{\min}[ H(\paramh_{N-1})])+ (\lambda_{\max}[\Delta\paramh H_1]  - \lambda_{\min}[\Delta\paramh H_1]) |\\
    \leq &  |\lambda_{\max}[ H(\paramh_{N-1})] - \lambda_{\min}[ H(\paramh_{N-1})]|+ |\lambda_{\max}[\Delta\paramh H_1]  - \lambda_{\min}[\Delta\paramh H_1]|\\
    = &\| H(\paramh_{N-1})\|_s + \|\Delta\paramh H_1\|_s \leq \| H(\paramh_{N-1})\|_s + |\Delta\paramh| \|H_1\|_s,
\end{align}
where in the first inequality we used that by definition $\lambda_{\max}[H(\paramh_N)]\geq \lambda_{\min}[H(\paramh_N)]$. 

We want to obtain a general condition that will work for an entire protocol. We can do so by taking the maximum/minimum over the different parameters at hand. Defining $ g_{\partial\max } := \max_l\partial_\paramh g_l(\paramh)$, the condition to fulfill is given by
\begin{gather}\label{eq:lower_bound_delta_paramh}
    (1-2\epsilon^2)\Delta_{\rm gap \, min} + |\Delta\paramh| M g_{\partial\max} (\|H\|_{s,\max} + \|H_1\|_{s,\max}|\Delta\paramh|)  + \|H_1\|_{s,\max}|\Delta\paramh|\leq 0\\
    \Downarrow\\
    |\Delta\paramh|<\frac{ \widetilde{\gamma} }{2} \left(\sqrt{\frac{(g_{\partial\max} \|H\|_{s,\max} M+\|H_1\|_{s,\max})^2+4 \Delta_{\rm gap \, min}  g_{\partial\max} M +\|H_1\|_{s,\max} \left(1-2 \epsilon
   ^2\right)}{(g_{\partial\max})^2 M^2 +\|H_1\|_{s,\max}^2}}-\frac{1}{g_{\partial\max} M}-\frac{\|H\|_{s,\max}}{\|H_1\|_{s,\max}}\right) ,
\end{gather}
where $\widetilde{\gamma}\in(0,1)$. We find a condition on $\Delta\paramh$ such that we can always ensure that the terms $\expval{  \left (H_{\rm eff}(\paramh_j) - P_1 H_{\rm eff}(\paramh_j) P_1 \right) }{\psi} $ in Eq.~\eqref{eq:lower_bound_variance_general} are always negative (and thus the product of the two is positive). 

Similarly to how we found $r$ in Eq.~\eqref{eq:condition_on_r_jleqN} $j\leq N-1$, we now find
\begin{align}\label{eq:condition_on_r_jeqN}
    r^2< \frac{3}{g_{\min}^2(M-1)}\frac{(1-\widetilde{\gamma})(1-2\epsilon^2)|\Delta_{\rm gap\, min}|}{\|H\|_{s,{\max}} + (1-\widetilde{\gamma})(1-2\epsilon^2)|\Delta_{\rm gap\, min}|}
\end{align} 

Note that Eq.~\eqref{eq:condition_on_r_jeqN} is simply a constrained version of Eq~\eqref{eq:condition_on_r_jleqN}, but it scales in exactly the same way with all relevant parameters, since $\widetilde{\gamma}$ is a constant. Because this condition yields a smaller admissible region in $r$ than the previous one, we will use it.

We have now shown that the expected values in Eq.~\eqref{eq:lower_bound_variance_general} are always negative within a certain region, and therefore the products over which we sum are always positive. This is almost sufficient to guarantee that the bound is strictly positive and hence non-trivial. In the next section, we analyze the function $h(\paramh_j,\paramh_k, r)$ in Eq.~\eqref{eq:lower_bound_variance_general} is positive. We now make the dependence on $r$ explicit, as it is the parameter that we study next.

\subsubsection{Positivity of the function \texorpdfstring{$h(\paramh_j,\paramh_k,r)$}{h(x_j,x_k,r)}}

Leveraging once again the Taylor Reminder Theorem~\ref{thm:taylor} we can show that
\begin{align}\label{eq:taylor_expand_h}
     h(\paramh_j,\paramh_k, r) = \frac{1}{4!}\partial_{r'}^4 h(\paramh_j,\paramh_k, r') \bigg|_{r' =0}r^4 + \frac{1}{6!}\partial_{r'}^6h(\paramh_j,\paramh_k, r')\bigg|_{r'=\nu\in[0,r]}r^6 
\end{align}
where we used that the 1st, 2nd and 3rd derivatives are zero, as well as the 5th derivative, which can be checked by direct calculation. We can find a lower bound on the function by considering the maximum (in $\nu$) absolute value of the $6$-th derivative with a negative sign up front. We do this by recalling the form of $h$ in Eq.~\eqref{eq:h_function} and then
\begin{align}
    \left| \partial_{r'}^6h(\paramh_j,\paramh_k, r')\bigg|_{r'=\nu\in[0,r]} \right|\leq &\frac{(2 g_{1}(\paramh_j))^6}{7} + \frac 12 (2 g_{1}(\paramh_j))^5 (2 g_{1}(\paramh_k)) + (2 g_{1}(\paramh_j))^4(2 g_{1}(\paramh_k))^2 + \frac 54 (2 g_{1}(\paramh_j))^3 (2 g_{1}(\paramh_k))^3 \\
    &+ (2 g_{1}(\paramh_k))^4(2 g_{1}(\paramh_j))^2 + \frac 12 (2 g_{1}(\paramh_k))^5 (2 g_{1}(\paramh_j)) +\frac{(2 g_{1}(\paramh_k))^6}{7}
    \\ 
    &+\frac 17 (2( g_{1}(\paramh_j)- g_{1}(\paramh_k)))^6+ \frac 17 (2( g_{1}(\paramh_j)+ g_{1}(\paramh_k)))^6    \\
    =& 2^6\Bigg( \frac{3  g_{1}(\paramh_j)^6}{7}+ \frac{ g_{1}(\paramh_j)^5  g_{1}(\paramh_k)}{2}+ \frac{37  g_{1}(\paramh_j)^4  g_{1}(\paramh_k)^2}{7}+ \frac{5 g_{1}(\paramh_j)^3 g_{1}(\paramh_k)^3}{4}\\
    &+ \frac{37  g_{1}(\paramh_j)^2  g_{1}(\paramh_k)^4}{7}+ \frac{ g_{1}(\paramh_j)  g_{1}(\paramh_k)^5}{7}+ \frac{ g_{1}(\paramh_j)  g_{1}(\paramh_k)^5}{2}+ \frac{3  g_{1}(\paramh_k)^6}{7}\Bigg) := h_6(\paramh_j,\paramh_k).
\end{align}
Here we computed the derivative using the chain rule and applied the well-known upper bound on the derivative of the $\sinc$ function, $|\partial_x^k\sinc(x)|<1/(1+k)$ (see Ref.~\cite{gronwall1920calculus} for details). To ease notation, we denote the resulting function by $h_6(\paramh_j,\paramh_k)$. 

Therefore, the only thing left to do is recover Eq~\eqref{eq:taylor_expand_h} and find the lower bound accordingly
\begin{align}
     h(\paramh_j,\paramh_k, r) \geq& \frac{1}{4!}\partial_{r'}^4 h(\paramh_j,\paramh_k, r') \bigg|_{r' =0}r^4 -\frac{1}{6!} h_6(\paramh_j,\paramh_k)r^6 \\
     =& \frac{4}{45}g_1^2(\paramh_j)g_1^2(\paramh_k) r^4 - h_6(\paramh_j,\paramh_k)r^6
\end{align}
by imposing this bound to be larger than zero we find that
\begin{gather}
    \frac{4}{45}g_1^2(\paramh_j)g_1^2(\paramh_k) r^4 - h_6(\paramh_j,\paramh_k)r^6>0\\
    \Downarrow\\
    r^2 <\frac{4 g_1^2(\paramh_j)g_1^2(\paramh_k)}{45 h_6(\paramh_j,\paramh_k)}\label{eq:second_condition_r_h}
\end{gather}

What remains is to combine all the conditions with the bound to obtain a final, concise expression.

\subsection{Final bound with all the necessary conditions.}

We have found two conditions on $r$, in Eqs~(\ref{eq:condition_on_r_jeqN},\ref{eq:second_condition_r_h}) respectively. We define a threshold $r_{\rm t}$ to be a range that fulfills both conditions, i.e. 
\begin{align}\label{eq:conditions_on_r}
    r_{\rm t}^2\leq \gamma \min\left\{\min_{j,k}\frac{4 g_1^2(\paramh_j)g_1^2(\paramh_k)}{45 h_6(\paramh_j,\paramh_k)}, \frac{3}{g_{\max}^2(M-1)}\frac{\widetilde{\gamma}(1-2\epsilon^2)|\Delta_{\rm gap\, min}|}{\|H\|_{s,{\rm max}} + (1-2\epsilon^2)|\Delta_{\rm gap\, min}|}\right\},
\end{align}
where again $\gamma\in(0,1)$. Note that both conditions are polynomial vanishing in the number of qubits under mild assumptions. Indeed, if the parameters of the problem are polynomial in the number of qubits, i.e. $g_{\max},M,\|H\|_{\rm max}\in\Theta({\rm poly}(n))$, $|\Delta_{\rm gap\, min}|\in\Theta(1/{\rm poly}(n))$, then the threshold conditions will also be vanishing polynomially at worst, i.e. $ r_{\rm t}\in\Theta(1/{\rm poly}(n))$.

To simplify the notation further down the line, we introduce the following definitions 
\begin{align}
    r_{\rm t_1} = & \gamma \min_{j,k}\frac{4 g_1^2(\paramh_j)g_1^2(\paramh_k)}{45 h_6(\paramh_j,\paramh_k)}\\
     r_{\rm t_2}^2 = & \gamma \frac{3}{g_{\min}^2(M-1)}\frac{\widetilde{\gamma}(1-2\epsilon^2)|\Delta_{\rm gap\, min}|}{\|H\|_{s,{\rm min}} + (1-2\epsilon^2)|\Delta_{\rm gap\, min}|}
\end{align}
denoting the two limiting thresholds for $r$ as $r_{\rm t_1},r_{\rm t_2}$.

Now we are equipped to recover Eq.~\eqref{eq:lower_bound_variance_general} and simplify the expression in the regime that $\alv\sim\DC(\vec{0},r_{\rm t})$
\begin{align}
    \Var_{\alv\sim\DC(\vec{0},r)}[\LC(\alv)]\geq & \frac{1}{N^2}\sum_{j,k}^N h(\paramh_j,\paramh_k) \expval{  \left (H_{\rm eff}(\paramh_j) - P_1 H_{\rm eff}(\paramh_j) P_1 \right) }{\psi} \\
    &\cdot \expval{  \left(H_{\rm eff}(\paramh_k) - 
 P_1 H_{\rm eff}(\paramh_k) P_1 \right)  }{\psi} \\
  \geq &\frac{1}{N^2}\min_{\paramh_{j'},\paramh_{k'}} h(\paramh_{j'},\paramh_{k'})\sum_{j,k}^N\expval{  \left (H_{\rm eff}(\paramh_j) - P_1 H_{\rm eff}(\paramh_j) P_1 \right) }{\psi}\\
  &\cdot \expval{  \left(H_{\rm eff}(\paramh_k) - 
 P_1 H_{\rm eff}(\paramh_k) P_1 \right) }{\psi} \\
 =&\frac{1}{N^2}\min_{\paramh_{j'},\paramh_{k'}} h(\paramh_{j'},\paramh_{k'}) \left( \sum_{j}^N \expval{  \left (H_{\rm eff}(\paramh_j) - P_1 H_{\rm eff}(\paramh_j) P_1 \right) }{\psi} \right)^2.\label{eq:lower_bound_var_toanalyze}
 \end{align}
In the first inequality, we just recovered Eq.~\eqref{eq:lower_bound_variance_general}. In the second inequality we used that all the terms in the sum are positive. Therefore,  the function $h$ can be pulled out of the sum, and we lower bound the expression by the minimum value o $h$ multiplied by the sum of the expectation-value terms. We also used that $\sum_{j,k}a_j a_k = \left(\sum_{j}a_j\right)^2$.

 Now we can analyze the two terms in Eq.~\eqref{eq:lower_bound_var_toanalyze} separately. First, we analyze the function $\min_{\paramh_{j'},\paramh_{k'}}h(\paramh_{j'},\paramh_{k'})$:
 \begin{align}
\min_{\paramh_{j'},\paramh_{k'}}h(\paramh_{j'},\paramh_{k'})\geq& r_{\rm t}^4\left( \frac{4}{45}g_1^2(\paramh_j)g_1^2(\paramh_k)  - h_6(\paramh_j,\paramh_k)r_{\rm t}^2\right)\\
\geq & r_{\rm t_{2}}^4\left( \frac{4}{45}g_1^2(\paramh_j)g_1^2(\paramh_k)  - h_6(\paramh_j,\paramh_k)r_{\rm t}^2\right)\in\Omega\left( 1/{\rm Poly(n)} \right),
 \end{align}
 where we used that $r_{\rm t}\in\Omega(1/{\rm Poly}(n))$ as all the terms in Eq.~\eqref{eq:conditions_on_r} are polynomial (and non zero). Furthermore, we used that $r_{\rm t_{2}} < r_{\rm t_{1}}$. We analyze further the term: 
\begin{align}\label{eq:scalability_of_the_hfunc}
     \left( \frac{4}{45}g_1^2(\paramh_j)g_1^2(\paramh_k)  - h_6(\paramh_j,\paramh_k)r_{\rm t}^2\right)
     \begin{cases}
         =(1-\gamma)\frac{4}{45}g_1^2(\paramh_j)g_1^2(\paramh_k)\in\Omega\left( 1/{\rm Poly}(n) \right) {\, \rm if \, }r_{\rm t} = r_{\rm t_1}\\
         >(1-\gamma)\frac{4}{45}g_1^2(\paramh_j)g_1^2(\paramh_k)\in\Omega\left( 1/{\rm Poly}(n) \right) {\, \rm if \, }r_{\rm t} = r_{\rm t_2},
     \end{cases}
 \end{align}
 where we used that if the threshold value is $r_{\rm t}=\gamma \min_{j,k}\frac{4g_1^2(\paramh_j)g_1^2(\paramh_k)}{45h_6(\paramh_j,\paramh_k)}$, then we can easily compute the value. Otherwise, it would imply that $r_{\rm t}$ is smaller than $\gamma \min_{j,k}\frac{4g_1^2(\paramh_j)g_1^2(\paramh_k)}{45h_6(\paramh_j,\paramh_k)}$, therefore, the quantity $\frac{4}{45}g_1^2(\paramh_j)g_1^2(\paramh_k)  - h_6(\paramh_j,\paramh_k)r_{\rm t}^2$ will be larger than in the previous case. 

 Let us now focus on the other term in Eq.~\eqref{eq:lower_bound_var_toanalyze}, i.e. the sum over the different expected values. Particularly, 
 \begin{align}
     \frac{1}{N^2}\left( \sum_{j}^N \expval{  \left (H_{\rm eff}(\paramh_j) - P_1 H_{\rm eff}(\paramh_j) P_1 \right) }{\psi} \right)^2 \geq& \frac{1}{N^2}\Bigg( \sum_{j=1}^{N-1} K_+^{(j)}(1-2\epsilon^2)\Delta_g^{(j)} + K_+^{(N)}\Big(\Delta_{\rm gap}^{(N-1)}(1-2\epsilon^2) \\
     &+ |\Delta\paramh| M \max_l|\partial_\paramh g_l(\paramh)| (\|H(\paramh_N)\|_s + |\Delta\paramh|\|H_1\|_s) + \|H_1\|_s|\Delta\paramh|\Big) \\
     &+ \sum_{j=1}^{N}(1-K_+^{(j)})\|H(\paramh_j)\|_s\Bigg)^2\\
     \geq & \frac{1}{N^2}\Bigg( \sum_{j=1}^{N-1} K_+^{(j)}(1-2\epsilon^2)\Delta_g^{(j)} + K_+^{(N)}\Big(\Delta_{\rm gap}^{(N-1)}(1-2\epsilon^2)\widetilde{\gamma} \\
     &+ \sum_{j=1}^{N}(1-K_+^{(j)})\|H(\paramh_j)\|_s\Bigg)^2\\
     \geq & \frac{1}{N^2}\Bigg( \sum_{j=1}^{N} K_+^{(j)}\widetilde{\gamma}(1-2\epsilon^2)\Delta_{\rm gap \, min} + \sum_{j=1}^{N}(1-K_+^{(j)})\|H(\paramh_j)\|_s\Bigg)^2\\
     \geq &  \Bigg(  \left( 1-\frac{g_{\max}^2(\paramh_j)r_{\rm t}^2(M-1)}{3}  \right)\widetilde{\gamma}(1-2\epsilon^2)\Delta_{\rm gap\, min} \\
     &+  \frac{g_{\max}^2(\paramh_j)r_{\rm t}^2(M-1)}{3} \|H_{\max}\|_s\Bigg)^2.
 \end{align}
In the first inequality we lower bounded this expected values by the bounds obtained before, specifically the expressions in Eqs.~(\ref{eq:lower_bound_expval_jleqN}, \ref{eq:bound_expval_jeqN}). Note that these equations provide upper bounds on negative quantities and therefore correspond to lower bounds on their absolute values. In the second inequality, we substitute the lower bound for $\Delta\paramh$ found in Eq.~\eqref{eq:lower_bound_delta_paramh}. In the second-to-last inequality we lower bounded the $\Delta_{\rm gap}^{(j)}$ by the minimum gap. Finally, in the last inequality we lower-bounded the $K_+$ with the lower-bound introduced in Eq.~\eqref{eq:lowerbound_on_K_wrt_r}. We can then analyze the resulting scaling as a function of $r_{\rm t}$exactly the same way as in Eq.~\eqref{eq:scalability_of_the_hfunc}, obtaining
\begin{align}\label{eq:scalability_of_the_expval}
       &\left|\widetilde{\gamma}(1-2\epsilon^2)\Delta_{\rm gap\, min} + \frac{g_{\max}^2(\paramh_j)r_{\rm t}^2(M-1)}{3}(\|H_{\max}\|_s- (1-\widetilde{\gamma})(1-2\epsilon^2)\Delta_{\rm gap\, min})\right|\\
       &\begin{cases}
       =  (1-\gamma)(1-\widetilde{\gamma})(1-2\epsilon^2)|\Delta_{\rm gap \, min}|\in \Omega(1/{\rm Poly}(n))\\
       \geq (1-\gamma)(1-\widetilde{\gamma})(1-2\epsilon^2)|\Delta_{\rm gap \, min}| \in \Omega(1/{\rm Poly}(n)).
       \end{cases}
 \end{align}
 We recall that  $\Delta_{\rm gap \, min}<0$, which is why we take the absolute value: the entire term is negative. Therefore, we can recover Eq.~\eqref{eq:lower_bound_var_toanalyze}, and lower bound it by polynomially vanishing terms. In particular, by combining the lower bounds on the two quantities that appear in the sum, as calculated in Eqs.~(\ref{eq:scalability_of_the_hfunc},\ref{eq:scalability_of_the_expval}), we obtain
\begin{align}
    \Var_{\alv\sim\DC(\vec{0},r)}[\LC(\alv)]\geq (1-\gamma)\frac{4 r_{{\rm t}_2}^4}{45}g_1^2(\paramh_j)g_1^2(\paramh_k)(1-\gamma)(1-\widetilde{\gamma})(1-2\epsilon^2)|\Delta_{\rm gap \, min}|\in\Omega\left( 1/{\rm Poly}(n) \right).
\end{align}

\end{proof}

\section{Synthetic shot-noise model}
\label{ap:shot_noise}

In this work we study training landscapes and scaling trends under finite-sampling conditions. On quantum hardware, expectation values are not accessed exactly, but rather they are estimated from a finite number of measurement shots, which introduces a stochastic error (shot noise). Simulating this effect is important if we want realistic loss fluctuations and fair comparisons across problem sizes and parameter counts at fixed measurement resources. However, explicitly sampling measurement outcomes at every optimization step adds a substantial computational overhead when we scale the number of qubits and repeat the procedure over many instances and training steps. To keep the study scalable while preserving the correct finite-shot statistics, we use a \emph{synthetic shot-noise model} that injects the corresponding fluctuations at the expectation-value level.

Concretely, our cost functions are expectation values of Hamiltonians that can be expressed as weighted sums of Pauli strings,
\begin{equation}
    H(\mathbf{x})=\sum_{\alpha=1}^{N_{\rm terms}} c_\alpha(\mathbf{x})\,P_\alpha,
    \qquad P_\alpha^2=\id,
\end{equation}
where $P_\alpha$ are Pauli operators, and $c_\alpha(x)$ are the Hamiltonian coefficients that may depend on $x$. Suppose we partition the index set into $G$ disjoint measurement groups
\(
\{\mathcal{I}_1,\ldots,\mathcal{I}_G\}
\)
(with $\mathcal{I}_g\cap \mathcal{I}_{g'}=\emptyset$ for $g\neq g'$ and $\cup_g \mathcal{I}_g=\{1,\dots,N_{\rm terms}\}$).
Define the grouped operators
\begin{equation}
    H_g(\mathbf{x}) := \sum_{\alpha\in \mathcal{I}_g} c_\alpha(\mathbf{x})\,P_\alpha,
    \qquad
    H(\mathbf{x})=\sum_{g=1}^G H_g(\mathbf{x}).
\end{equation}
For a (pure or mixed) state $\rho(\boldsymbol{\theta},\mathbf{x})$ the exact loss is
\begin{equation}
    \mathcal{L}(\boldsymbol{\theta};\mathbf{x})
:=\Tr\!\big[\rho(\boldsymbol{\theta},\mathbf{x})\,H(\mathbf{x})\big]
    =\sum_{g=1}^G \mathcal{L}_g(\boldsymbol{\theta};\mathbf{x}),
    \qquad
\mathcal{L}_g(\boldsymbol{\theta};\mathbf{x})
:=\Tr\!\big[\rho(\boldsymbol{\theta},\mathbf{x})\,H_g(\mathbf{x})\big], 
\end{equation}
where each group $g\in\{1,\dots,G\}$ collects Pauli terms that are estimated using the same measurement setting (i.e., the same choice of local measurement bases, or equivalently the same ``measurement circuit'' up to single-qubit basis changes). In a shot-based implementation, the contributions $\mathcal{L}_g$ are therefore obtained from independent batches of measurements, with a group-dependent shot budget $S_g$ (or, more generally, an allocation $\{S_\alpha\}_{\alpha\in\mathcal{I}_g}$ within the group). This measurement-group structure is the reason we model finite-shot fluctuations at the \emph{group level}: for each evaluation of $\mathcal{L}_g(\boldsymbol{\theta};\mathbf{x})$ we draw a single random fluctuation $\xi_g$ and define a synthetic estimator
\begin{equation}
    \widehat{\mathcal{L}}_g(\boldsymbol{\theta};\mathbf{x})
    :=
    \mathcal{L}_g(\boldsymbol{\theta};\mathbf{x})+\xi_g,
    \qquad
    \xi_g\sim\mathcal{N}\!\big(0,V_g(\boldsymbol{\theta};\mathbf{x})\big),
\end{equation}
where $V_g(\boldsymbol{\theta};\mathbf{x})$ is chosen to reproduce the leading-order shot-noise variance of the group estimator. Under the standard assumption that different terms in the group are estimated from independent shot sets (or that intra-group covariances are negligible), one has
\begin{equation}
    V_g(\boldsymbol{\theta};\mathbf{x})
    \approx
    \sum_{\alpha\in\mathcal{I}_g} c_\alpha(\mathbf{x})^2\,
    \Var\!\big[\widehat{\langle P_\alpha\rangle}_{\boldsymbol{\theta},\mathbf{x}}\big]
    =
    \sum_{\alpha\in\mathcal{I}_g} c_\alpha(\mathbf{x})^2\,
    \frac{1-\langle P_\alpha\rangle_{\boldsymbol{\theta},\mathbf{x}}^{\,2}}{S_\alpha},
    \label{eq:group_var_general_app}
\end{equation}
with $S_\alpha$ the number of shots allocated to term $P_\alpha$ (e.g., $S_\alpha=S_g/|\mathcal{I}_g|$ for uniform allocation). Finally, drawing $\{\xi_g\}_{g=1}^G$ independently across groups (reflecting independent measurement settings), we define the synthetic shot-noisy loss
\begin{equation}
    \widehat{\mathcal{L}}(\boldsymbol{\theta};\mathbf{x})
    :=
    \sum_{g=1}^G \widehat{\mathcal{L}}_g(\boldsymbol{\theta};\mathbf{x})
    =
    \mathcal{L}(\boldsymbol{\theta};\mathbf{x})+\sum_{g=1}^G \xi_g,
\end{equation}
which satisfies $\mathbb{E}[\widehat{\mathcal{L}}]=\mathcal{L}$ and
$\Var[\widehat{\mathcal{L}}]\approx \sum_{g=1}^G V_g$ by construction.

Equivalently, one may write the group estimator in the ``standard normal'' form
\begin{equation}
    \widehat{\mathcal{L}}_g(\boldsymbol{\theta};\mathbf{x})
    =
    \mathcal{L}_g(\boldsymbol{\theta};\mathbf{x})
    + \sqrt{V_g(\boldsymbol{\theta};\mathbf{x})}\,\zeta_g,
    \qquad
    \zeta_g\sim\mathcal{N}(0,1),
\end{equation}
which is identical to the previous expression since $\sqrt V_g\zeta_g\sim\mathcal{N}(0,V_g)$.
In practice, each time the optimizer requests a loss evaluation (either to compute $\widehat{\mathcal{L}}$ itself or as part of a gradient-estimation rule such as parameter shift), we draw independent $\{\zeta_g\}_{g=1}^G$. This matches the shot-based setting, where each circuit evaluation corresponds to an independent batch of measurement outcomes. Importantly, gradients are obtained by applying the same finite-shot procedure to the noisy loss evaluations (e.g., combining two independently noised shifted losses in the parameter-shift rule), so that the resulting gradient estimator remains unbiased, $\mathbb{E}[\widehat{\nabla_{\boldsymbol{\theta}}\mathcal{L}}]=\nabla_{\boldsymbol{\theta}}\mathcal{L}$, rather than by differentiating through a fixed realization of the random fluctuations.

\end{document}